\documentclass[11pt]{article}
\usepackage[utf8]{inputenc}
\usepackage{amsmath,amsthm,amsfonts,amssymb,amscd}

\usepackage{thm-restate}

\newcommand\Algphase[1]{%
\vspace*{.5\baselineskip}%
\nonl \hspace*{-0em}\textbf{\large #1}%
\vspace*{.01\baselineskip}%
}

\usepackage{sn-preamble}
\usepackage{bbm}
\usepackage{framed}
\usepackage[most]{tcolorbox}
\colorlet{shadecolor}{orange!15}

\usepackage{tikz}
\usepackage{verbatim}
\usepackage{cancel}
\usepackage{caption}

\usetikzlibrary{decorations.pathreplacing,angles,quotes}
\usetikzlibrary{shapes.geometric}
\usetikzlibrary{patterns}

\usepackage{mathtools} %
\theoremstyle{definition}

\DeclareMathAlphabet{\mathfrak}{U}{jkpmia}{m}{it}
\SetMathAlphabet{\mathfrak}{bold}{U}{jkpmia}{bx}{it}

\def\Apr{\textsc{Approx}}
\def\RelV{\textsc{FindVarValue}}
\def\bI{\boldsymbol{I}}

\def\MapBack{\textsc{MapBack}}

\def\UniformJunta{\textsc{UniformJunta}}

\def\sfg{\mathsf{G}}

\let\oldnl\nl
\newcommand{\nonl}{\renewcommand{\nl}{\let\nl\oldnl}}

\begin{document}
\setcounter{section}{0}

\title{Testing Juntas and Junta Subclasses with Relative Error}

\newcommand{\red}[1]{{\color{red} {#1}}}
\newcommand{\blue}[1]{{\color{blue} {#1}}}
\newcommand{\gray}[1]{{\color{gray} {#1}}}

\newcommand{\xnote}[1]{\footnote{{\bf \color{purple}Xi}: {#1}}}
\newcommand{\rnote}[1]{\footnote{{\bf \color{red}Rocco}: {#1}}}
\newcommand{\toni}[1]{\footnote{{\bf\color{blue} toni}: {#1}}} 
\newcommand{\wnote}[1]{\footnote{{\bf \color{teal}Will}: {#1}}}

\newcommand{\bxO}{\mathbf{x^1}}
\newcommand{\bxT}{\mathbf{x^2}}
\newcommand{\bxTh}{\mathbf{x^3}}

\newcommand{\byO}{\mathbf{y^1}}
\newcommand{\byT}{\mathbf{y^2}}
\newcommand{\byTh}{\mathbf{y^3}}

\newcommand{\bGamma}{\mathbf{\Gamma}}

\newcommand{\uth}{\bigskip \bigskip{\huge{\bf UP TO HERE} \bigskip \bigskip}}
\newcommand{\SAMP}{{\bf Samp}}
\newcommand{\MQ}{\mathrm{MQ}}

\newcommand{\reldist}{\mathrm{rel}\text{-}\mathrm{dist}}
\newcommand{\default}{\ensuremath{\mathsf{default}}}

\author{
Xi Chen\thanks{Columbia University. Email: \url{xc2198@columbia.edu}.} \and 
William Pires \thanks{Columbia University. Email: \url{wp2294@columbia.edu}.} \and 
Toniann Pitassi\thanks{Columbia University. Email: \url{tonipitassi@gmail.com}.} \and 
Rocco A. Servedio\thanks{Columbia University. Email: \url{rocco@cs.columbia.edu}.} 
}

\newcommand{\CApr}{\mathcal{C}(k)\textsc{-Approximators}}
\newcommand{\fcj}{\mathfrak{j}}
\newcommand{\fA}{\mathfrak{A}}
\newcommand{\fJ}{\mathfrak{J}}
\newcommand{\fF}{\mathfrak{F}}
\newcommand{\fG}{\mathfrak{G}}

\thispagestyle{empty}

\maketitle

\begin{abstract}%
This papers considers the junta testing problem in a recently introduced ``relative error'' variant of the standard Boolean function property testing model.
In relative-error testing we measure the distance from $f$ to $g$, where $f,g: \zo^n \to \zo$, by the ratio of $|f^{-1}(1) \hspace{0.06cm} \triangle \hspace{0.06cm} g^{-1}(1)|$ (the number of inputs on which $f$ and $g$ disagree) to $|f^{-1}(1)|$ (the number of satisfying assignments of $f$), and we give the testing algorithm both black-box access to $f$ and also  access to independent uniform samples from $f^{-1}(1)$.  

\cite{CDHLNSY2024} observed that the class of $k$-juntas is $\poly(2^k,1/\eps)$-query testable in the relative-error model, and asked whether $\poly(k,1/\eps)$ queries is achievable.  We answer this question affirmatively by giving a $\tilde{O}(k/\eps)$-query algorithm, matching the optimal complexity achieved in the less challenging standard model.  Moreover, as our main result, we show that any \emph{subclass} of $k$-juntas that is closed under permuting variables is relative-error testable with a similar complexity.  This gives highly efficient relative-error testing algorithms for a number of well-studied function classes, including size-$k$ decision trees, size-$k$ branching programs, and size-$k$ Boolean formulas.
\end{abstract}

\def\frakg{{\mathfrak g}}

\section{Introduction} \label{sec:intro}

A Boolean function $f: \zo^n \to \zo$ is a \emph{$k$-junta} if it depends on at most $k$ of its $n$ input variables.
The class of $k$-juntas arises across many disparate contexts in theoretical computer science: for example, ``junta theorems'' are of fundamental importance in the analysis of Boolean functions \cite{NisanSzegedy:92,Friedgut:98,Bourgain,KindlerSafra04,Filmus22}, ``junta learning'' is a paradigmatic problem capturing the difficulty of identifying relevant variables in learning theory \cite{Blum:94,blu03,MOS:03,Valiant15}, and, closest to our concerns, juntas play a central role in the field of \emph{property testing}.

Starting with the influential work of Fischer et al.~\cite{FKRSS03}, testing whether an unknown and arbitrary function $f: \zo^n \to \zo$ is a $k$-junta has emerged as one of the touchstone problems in Boolean function property testing. In the standard testing model introduced by Blum, Luby and Rubinfeld \cite{BLR93,RS96}, junta testing has been intensively studied from many perspectives, with upper and lower bounds for both adaptive and non-adaptive junta testing algorithms given in a long sequence of papers
\cite{ChocklerGutfreund:04,FKRSS03,Blais08,Blaisstoc09,BGMW13,STW15,Saglam18,CSTWX18jacm}.
Junta testing has also been studied in several more recently introduced and challenging property testing frameworks that have been proposed over the years. These include the model of \emph{distribution-free} property testing \cite{HalevyKushilevitz:07,LCSSX19,Bshouty19,belovs2019quantum}, the \emph{tolerant testing} model \cite{DMN19,ITW21,PRW22,LW19,ChenPatel23,CDLNS24,NadimpalliPatel24}, \emph{quantum} property testing models \cite{AticiServedio:07qip,ABRW16,belovs2019quantum,CNY23} and the model of \emph{active property testing} \cite{BBBY12,Alon12}.

The contribution of this paper is to resolve (a generalization of) the junta testing problem in the recently introduced \emph{relative-error} model of Boolean function property testing \cite{CDHLNSY2024,CPPS25,CDHNSY25}.  Our main result is an efficient relative-error testing algorithm for $k$-juntas, and more generally for any subclass of $k$-juntas that is closed under permutations of the variables, resolving an open question posed in \cite{CDHLNSY2024}.  Our algorithm for testing subclasses of juntas yields efficient relative-error testing algorithms for a number of natural function classes, including: size-$s$ decision trees, size-$s$ branching programs, and size-$s$ Boolean formulas, whose standard-model and distribution-free testability was studied in \cite{DLM+:07,CGM11,Bshouty20}.

\subsection{Prior results on relative-error testing}

The relative-error property testing model for Boolean functions was introduced in \cite{CDHLNSY2024}  to provide an appropriate framework for testing \emph{sparse} Boolean functions (ones with few satisfying assignments).  Recall in the standard Boolean function property testing model,\footnote{See \Cref{sec:prelims} for detailed definitions of both this model and the relative-error testing model. 
}
the distance between functions $f,g: \zo^n \to \zo$ is the normalized Hamming distance ${  {|f^{-1}(1) \ \triangle \ g^{-1}(1)| }\big/{2^n}}$, and hence any $f$ that has fewer than $\eps 2^n$ satisfying assignments will trivially be $\eps$-close to the constant-0 function. Motivated by graph property testing models introduced in \cite{GoldreichRon02,ParnasRon02} which are suitable for \emph{sparse} graphs, \cite{CDHLNSY2024} defined the \emph{relative-error} Boolean function property testing model to measure the distance between the function $f: \zo^n \to \zo$ that is being tested and any other function $g$ to be
$
\reldist(f,g) := { {|f^{-1}(1) \hspace{0.05cm} \triangle \hspace{0.05cm} g^{-1}(1)|}\big/{|f^{-1}(1)|}}.
$ With this definition, $g$ is close to $f$ only if the symmetric difference is small ``at the scale of $f$,'' and it is reasonable to consider testing sparse Boolean functions. The model also allows the testing algorithm to obtain i.i.d.~uniform elements of $f^{-1}(1)$ by calling a ``random sample'' oracle (since if only black-box $\MQ(f)$ queries were allowed, as in the standard model, then for sparse $f$ a very large number of queries would be required to obtain any information about $f$ at all).

\cite{CDHLNSY2024} established some general results about relative-error testing vis-a-vis standard-model property testing.  They showed that any class of functions ${\cal C}$ can be tested in the standard model with essentially no more queries than are required for relative-error testing (so relative-error testing is always at least as hard as standard testing), and gave a simple example of an (artificial) class ${\cal C}'$ for which relative-error testing algorithms require $2^{\Omega(n)}$ queries but standard-model testers require only constantly many queries (so relative-error testing can be strictly harder than standard testing).  However, the chief focus of \cite{CDHLNSY2024} was on relative-error \emph{monotonicity testing}.  The main results of \cite{CDHLNSY2024} showed that the complexity of relative-error monotonicity testing is within a polynomial factor of the complexity of standard-model monotonicity testing (see the first row of \Cref{table:results}).

In subsequent work, \cite{CDHNSY25} studied relative-error testing of \emph{halfspaces}, and gave an $\tilde{\Omega}(\log n)$ lower bound on relative-error halfspace testing algorithms. In contrast, there are  standard-model halfspace testing algorithms with query complexity $\poly(1/\eps)$ (independent of $n$) \cite{MORS:10}, thus showing that the natural class of halfspaces witnesses an arbitrarily large gap between the query complexity in the relative-error model and the standard model. On the other hand, \cite{CPPS25} showed that for two other natural classes of Boolean functions, namely \emph{conjunctions} and \emph{decision lists}, there are $\tilde{O}(1/\eps)$-query relative-error testing algorithms, matching the performance of the best standard-model testers \cite{PRS02,Bshouty20}.
See \Cref{table:results} for a summary of how relative-error testing results contrast with standard-model testing results for various well-studied function classes.

\subsection{Relative-error junta testing and our results}

\cite{CDHLNSY2024} posed a number of questions for future work on relative-error property testing, the first of which concerns junta testing.
We recall that in the standard model,  \cite{Blaisstoc09} gave an $\eps$-testing algorithm for $k$-juntas that makes $O(k \log k +  k/\eps)$ queries, and \cite{Saglam18} subsequently established an $\Omega(k \log k)$ lower bound, showing that the algorithm of \cite{Blaisstoc09} is optimal up to constant factors for constant $\eps$.
\cite{CDHLNSY2024} made the simple observation that if every nonzero function in a class ${\cal C}$ has at least $p2^n$ satisfying assignments, then relative-error testing to accuracy $\eps$ reduces to standard-model testing to accuracy $p\eps$.  Since every nonzero $k$-junta has $p \geq 2^{-k}$, combining this with the standard-model junta testing result of \cite{Blaisstoc09} implies that $k$-juntas can be relative-error tested with a query complexity of $O(2^k/\eps)$.  As their first open question, \cite{CDHLNSY2024} asked whether there is a $\poly(k,1/\eps)$-query relative-error tester for $k$-juntas.

\begin{table}  
\centering
\renewcommand{\arraystretch}{1.34}
  \begin{tabular}{ @{}p{0.25\textwidth}p{0.31\textwidth}p{0.30\textwidth}@{} }
 \toprule
Class of functions & Standard-model & Relative-error \\[-0.5em]
& testing complexity  & testing complexity\\
\midrule 
Monotone  functions   & $\tilde{O}(\sqrt{n}/\eps^2)$ \cite{KMS18},& $O(n/\eps)$ \cite{CDHLNSY2024}\\[-0.25em]
&   $\tilde{\Omega}(n^{1/3})$ \cite{CWX17stoc} & $\tilde{\Omega}(n^{2/3})$ \cite{CDHLNSY2024}\\
 \hline
Halfspaces  & $\poly(1/\eps)$ \cite{MORS10}& $\tilde{\Omega}(\log n)$ \cite{CDHNSY25} \\
 \hline
Conjunctions & $\Theta(1/\eps)$ \cite{PRS02} & $\Theta(1/\eps)$ \cite{CPPS25}   \\
 \hline
Decision lists  & $\tilde{\Theta}(1/\eps)$ \cite{Bshouty20} & $\tilde{\Theta}(1/\eps)$ \cite{CPPS25} \\
 \hline
$k$-juntas  & $\tilde{O}(k/\eps)$ \cite{Blaisstoc09}, & $\tilde{O}(k/\eps)$ {\textbf{[This paper]}} \\[-2.2em]
 &  \flushleft $\tilde{\Omega}(k)$ \cite{ChocklerGutfreund:04,Saglam18} &  ~ \\
\hline
Subclass ${\cal C}(k)$ of $k$-juntas & $\tilde{O}\pbra{{\frac {k + \log|{\cal C}(k)^*|}\eps}}$ \cite{Bshouty20} & 
$\tilde{O}\pbra{{\frac {k\log|{\cal C}(k)^*|}\eps}}$  {\textbf{[This paper]}} \\
\hline
size-$k$ decision trees & $\tilde{O}(k/\eps)$ \cite{Bshouty20}  & $\tilde{O}(k^2/\eps)$ {\textbf{[This paper]}}\\
\hline
size-$k$ branching programs & $\tilde{O}(k/\eps)$ \cite{Bshouty20}  & $\tilde{O}(k^2/\eps)$ {\textbf{[This paper]}}\\
\hline
size-$k$ Boolean formulas & $\tilde{O}(k/\eps)$ \cite{Bshouty20}  & $\tilde{O}(k^2/\eps)$ {\textbf{[This paper]}}\\
\hline
\bottomrule
\end{tabular}
\caption{Known results on the number of oracle calls needed for standard-model $\eps$-testing and relative-error $\eps$-testing of various classes of $n$-variable Boolean functions.  For relative-error testing, the entry indicates the total number of oracle calls to either the random example oracle or the $\MQ$ oracle. As explained in \Cref{sec:subclass}, for our results ${\cal C}(k)$ is an arbitrary subclass of $k$-juntas that is closed under permuting input variables, and ${\cal C}(k)^* \subset  {\cal C}(k)$ is the subset of ${\cal C}(k)$ consisting of the functions whose set of relevant variables is contained in $\{1,\dots,k\}$. 
}
\label{table:results}
\end{table}

Our first result answers this question in the affirmative: 

\begin{theorem} [Relative-error junta testing] \label{thm:junta-main}
For $0 < \eps < 1/2$, there is an algorithm for relative-error $\eps$-testing of $k$-juntas with $O(k/\eps\log(k/\epsilon))$ queries and samples.
\end{theorem}

Our main result goes beyond juntas and deals with \emph{subclasses} of juntas.\footnote{We remind the reader that the testability/non-testability of a class ${\cal C}_1$ does not imply anything about the testability/non-testability of a class ${\cal C}_2 \subseteq {\cal C}_1$. This is true in both the standard model and in the relative-error testing model:  for example, the class of all Boolean functions is trivially relative-error testable with zero queries, but the subclass of halfspaces is ``hard'' to test, requiring $\smash{\tilde{\Omega}(\log n)}$ queries as shown by \cite{CDHNSY25}. In the other direction, conjunctions are a subclass of halfspaces, but while halfspaces are ``hard'' to test conjunctions are ``easy,'' requiring only $O(1/\eps)$ queries as shown by \cite{CPPS25}.)}  
To state the result cleanly, the following terminology and notation is helpful:  We say that a class ${\cal C}(k)$ of Boolean functions from $\zo^n$ to $\zo$ is a \emph{subclass of $k$-juntas} if (i) every $f \in {\cal C}(k)$ is a $k$-junta, and (ii) ${\cal C}(k)$ is closed under permuting variables.
We further write ${\cal C}(k)^*$ to denote the subset 
\[
{\cal C}(k)^* := \{f \in {\cal C}(k) : 
\text{~$f$ does not depend on variables $x_{k+1},\dots,x_n$}\}
\]
so we may think of ${\cal C}(k)^*$ as a class of functions over variables $x_1,\dots,x_k$; note that consequently the cardinality $|{\cal C}(k)^*|$ is completely independent of $n$ and depends only on $k$.

By combining techniques from junta testing with ideas from computational learning theory, \cite{DLM+:07} showed that any subclass ${\cal C}(k)$ of $k$-juntas can be tested using only $\poly(\log(|{\cal C}(k)^*|),1/\eps)$ queries in the standard model.  The precise polynomial bounds achieved by \cite{DLM+:07} were subsequently improved by \cite{CGM11} and by  \cite{Bshouty20}.  Moreover, \cite{Bshouty20} extended these result to the \emph{distribution-free} testing model; we will describe this model, and the relevance of Bshouty's approach to our work, in \Cref{sec:techniques}.

Our second and main result shows that any %
$\calC(k)$ can be efficiently tested in the relative-error model:

\begin{theorem} [Relative-error testing of junta subclasses] \label{thm:junta-subclass-main}
Let ${\cal C}(k)$ be any subclass of $k$-juntas.  For $0 < \eps < 1/2$, there is an algorithm for its relative-error $\eps$-testing %
with $\tilde{O}\pbra{{ {k\log|{\cal C}(k)^*|}/\eps}}$  queries and samples.
\end{theorem}

(We remark that since there are $2^{2^k}$ juntas over variables $x_1,\dots,x_k$, \Cref{thm:junta-main} does not follow in a black-box way from \Cref{thm:junta-subclass-main}; the proof of \Cref{thm:junta-main} requires a more refined analysis.)

Since any size-$k$ decision tree is a $k$-junta, and likewise for size-$k$ Boolean formulas and size-$k$ branching programs, as a  consequence of \Cref{thm:junta-subclass-main} via standard counting arguments given in \cite{DLM+:07} we get the following relative-error analogues of standard-model testing results first by~\cite{DLM+:07} (and subsequently with sharper parameters by \cite{CGM11,Bshouty20}):

\begin{corollary} \label[corollary]{cor:subclass}
Let ${\cal C}(k)$ be the class of size-$k$ decision trees, or the class of size-$k$ Boolean formulas, or the class of size-$k$ branching programs.  For $0 < \eps <1/2$, there is an algorithm for relative-error $\eps$-testing of ${\cal C}(k)$ that makes $\tilde{O}\pbra{k^2/\eps}$ queries and samples.
\end{corollary}

\subsection{Techniques} \label{sec:techniques}

In this subsection we focus our discussion on the more general and challenging problem of testing subclasses of juntas (\Cref{thm:junta-subclass-main}). 
(An overview of \Cref{thm:junta-main} is given in  \Cref{sec:juntas}.)
Our approach builds on techniques that were used to give efficient algorithms for testing juntas and subclasses of juntas in the \emph{distribution-free} testing model.
We recall that in the distribution-free testing model, as defined by \cite{GGR98} and subsequently studied by many authors \cite{HalevyKushilevitz:07,HalevyKushilevitz:08,GlasnerServedio:09toc,DolevRon11,CX16,LCSSX19,Bshouty19,belovs2019quantum,CP22,CFP24},
there is an unknown and arbitrary ``background distribution'' ${\cal D}$ over $\zo^n$, and the distance between two functions $f$ and $g$ is measured by $\Pr_{\bx \sim {\cal D}}[f(\bx) \neq g(\bx)].$  In this model the testing algorithm can both make black-box queries to the unknown $f$ that is being tested, and can also draw i.i.d.~samples from the distribution ${\cal D}$.

\medskip

\begin{remark} [On distribution-free testing versus relative-error testing.] \label[remark]{rem:dist-free-versus-relative}
We stress that while a relative-error testing algorithim gets access to independent uniform samples from $f^{-1}(1)$, akin to how a distribution-free tester gets access to independent samples from ${\cal D}$, relative-error testing does \emph{not} correspond to the special case of distribution-free testing in which the ``background distribution'' ${\cal D}$ is uniform over $f^{-1}(1)$.  This is because  in the distribution-free model, the distance $\Pr_{\bx \sim {\cal D}}[f(\bx) \neq g(\bx)]$ between two functions $f,g$ is measured vis-a-vis the unknown distribution ${\cal D}$ which the testing algorithm can draw samples from.  This is quite different from the way distance is measured in the relative-error setting; in particular, the relative distance $\reldist(f,g)$ is \emph{not} equal to $\smash{\Pr_{\bx \sim {\cal U}_{f^{-1}(1)}}[f(\bx)\neq g(\bx)]}$. Indeed, a function $g$ can have large relative distance from $f$ because $g(x)=1$  on many points $x$ for which $f(x)=0$, but all such points have zero weight under the distribution ${\cal U}_{f^{-1}(1)}$.  This issue necessitates a number of changes to the techniques used to give prior distribution-free testers, as described below.
\end{remark}

Improving on \cite{LCSSX19}, \cite{Bshouty19} gave a distribution-free testing algorithm that uses 
$\tilde{O}(k/\eps)$ 
queries and samples, and soon thereafter \cite{Bshouty20}, extending and strengthening the learning-based approach of \cite{DLM+:07}, gave an efficient distribution-free tester for any subclass of $k$-juntas.  
At a high level our approach is similar to that of \cite{Bshouty20}, but there are several significant differences (most significant for the last two phases of our algorithm), which we describe in the high-level sketch of our approach given below.

\begin{itemize}

\item 
Like \cite{Bshouty20}, the first two phases of our main {\textsc{Junta-Subclass-Tester}} algorithm (\Cref{alg:junta-subclass-tester}) checks that certain  restrictions of the function $f:\{0,1\}^n\rightarrow \{0,1\}$ that is being tested are close to literals in a suitable sense.  This is done in our algorithm as follows: in the first phase (``Partition $[n]$ and find relevant blocks''), the algorithm randomly partitions the $n$ input variables $x_1,\dots,x_n$ into $O(k^2)$ disjoint blocks, and performs a binary search procedure over blocks to find (i) a collection of ``relevant blocks,'' and (ii) for each ``relevant block,'' a string (denoted $v^{\ell}$) that ``witnesses'' that the block is relevant. 
(The whole algorithm rejects if this phase identifies more than $k$ such blocks.)
Next, the second phase (``Test the relevant sets'') uses a uniform-distribution 1-junta testing algorithm \UniformJunta~due to \cite{Blais08} 
(see \Cref{thm: uniform junta blais} in \Cref{sec:background-prop-test}) to check that certain restrictions of $f$ (denoted $f^\ell$, and defined using the witness strings $v^\ell$) are suitably close to  literals under the uniform distribution. 

At a high level these phases of our algorithm are similar to \cite{Bshouty20}, but the differences between the relative-error framework and the distribution-free framework mean we must carry out these phases in a somewhat different way.
\cite{Bshouty20} creates restrictions and uses witness strings which, roughly speaking, fix 
variables which ``appear to be irrelevant'' to 0.\footnote{\cite{Bshouty20}'s results only apply to subclasses of juntas that are closed under variable projections and
under restrictions that fix variables to 0.
}
This approach works in the distribution-free context, since %
roughly speaking, 
the distance between two functions can only come from examples $x$ that are ``likely to be drawn under ${\cal D}$.'' In contrast, in the relative-error setting, 
a function $g$ can be very far from $f$ because $g$ has many satisfying assignments $x$ for which $f(x)=0$, and such points will never be drawn when we sample from $f$'s satisfying assignments.  Because of this, we need to proceed in a different way than \cite{Bshouty20}; as detailed in \Cref{sec:subclass}, our approach is to fix the ``seemingly irrelevant'' variables to \emph{uniformly random} values.

\item In the third phase (``Trimming $\Apr$''), our algorithm creates a set 
of  candidate functions which are possible ``approximators'' (in a suitable sense) for the unknown function $f$. 
Recall from \Cref{table:results} that ${\cal C}(k)^* \subset {\cal C}(k)$ is the subset of ${\cal C}(k)$ consisting of the functions in ${\cal C}(k)$ whose set of relevant variables is contained in $\{1,\dots,k\}$;  
our algorithm carries out the third phase by ``trimming'' an initial set $\Apr$ which
consists of a specially constructed set of juntas that each have small relative distance to some function in ${\cal C}(k)^*$. 
This ``trimming'' is done by checking, for each candidate function $g$ in $\Apr$, whether randomly sampled positive examples of $f$ also correspond (after a suitable remapping of variables) to positive examples of the candidate function $g$; if this does not hold for one of the sampled positive examples, then the candidate $g$ is discarded. (The algorithm rejects if all candidates in $\Apr$ are discarded.)

In its fourth and final phase (``Find $k$-variable approximator for $f$''), our algorithm determines whether %
the trimmed approximator set contains a suitable approximator for $f$.  This is done by (i) identifying the function $\frakg$ in 
the trimmed approximator set %
which has the fewest satisfying assignments,
and (ii) verifying that 
$\frakg$ is indeed close to $f$ in a suitable sense (under 
a remapping of the variables). If this is the case then our algorithm accepts,  %
and otherwise it rejects.

The most significant differences between our algorithm and that of \cite{Bshouty20} are in Phases~3 and~4.
\cite{Bshouty20} starts by checking that $f$ is close to a particular subfunction of $f$, denoted $F$, which sets all variables in ``irrelevant'' blocks of $f$ to 0 and sets all variables in each ``relevant'' block of $f$ to the same bit.  
\cite{Bshouty20} then checks  $F$ against a set of candidates, each of which belongs to the subclass ${\cal C}(k)$.
Because the \cite{Bshouty20} (distribution-free) notion of two functions $h,h'$ being $\eps$-far is that $h(\bx) \neq h'(\bx)$ with probability at least $\eps$ over $\bx \sim {\cal D}$, it is enough to simply sample independent draws from ${\cal D}$ and reject if $f$ and $F$ disagree on any of them. Similarly for the second check of $F$ against the candidates, if $F$ and a candidate disagree on a sampled point $\bx$ then the candidate is eliminated.

In contrast, in our relative-error setting, since we can only draw samples from $f^{-1}(1)$, after we trim the approximator set in Phase 3, it could still be possible that $f$ is far in relative distance from some surviving $g$ in the trimmed set, which could happen if $g^{-1}(1) \cap f^{-1}(0)$ is large. 

To deal with this, our approach is as follows: we do not use an analogue of \cite{Bshouty20}'s function $F$, but rather we  consider {\it only} the function $\frakg$ in the trimmed set
that has the \emph{fewest} satisfying assignments, and we work solely with this function in Phase~4.
In the yes-case, when $f\in \mathcal{C}(k)$, we must have $|\frakg^{-1}(1)| \approx |f^{-1}(1)|$, and since $\frakg(\bx)=1$ for $\bx \sim f^{-1}(1)$ with high probability, it must also be the case that $f(\bx)=1$ for $\bx \sim \frakg^{-1}(1)$ with high probability.  And in the no case, when $\reldist(f,h)$ is large for every $h \in \mathcal{C}(k)$, we cannot have both $\frakg(\bx)=1$ for $\bx \sim f^{-1}(1)$ with high probability and also  $f(\bx)=1$ for $\bx \sim \frakg^{-1}(1)$ with high probability. For this would imply that $\reldist(f,\frakg)$ is small, but since $\frakg$ is very close in relative distance to some function in $\mathcal{C}(k)$, a simple triangle inequality
implies that $f$ has small relative distance to some function in ${\cal C}(k)$, which contradicts the fact that we are in the no-case. 
\end{itemize}

\subsection{Future work:  relative-error testing of classes of functions that are approximated by juntas?}

A natural goal for future work is to obtain positive results for testing classes of functions that do not correspond exactly to subclasses of juntas. In particular, in the standard uniform-distribution model, the techniques of \cite{DLM+:07} go beyond just testing subclasses of juntas and allow testing classes of functions which can be \emph{approximated} to high accuracy by juntas under the uniform distribution. Several classes of interest are of this sort, including: size-$s$ Boolean circuits (where circuit size is measured by the number of unbounded fan-in AND/OR/NOT gates), $s$-term DNF formulas, $s$-sparse polynomials over $\F_2^n$, and others; each of these  can be approximated, but not exactly computed, by juntas.  Can classes such as these be tested efficiently in the relative-error setting?  

A challenge which immediately arises is that while these classes are well-approximated by juntas under the uniform distribution, this is not the case under relative error.  For example, for any constant $s$ it is easy to give an $s$-term DNF formula which cannot be well-approximated in relative error by any $O_s(1)$-junta.  On the positive side, though, for both the class ${\cal C}=$ conjunctions and the class ${\cal C} =$ decision lists, functions in ${\cal C}$ are approximable by $O(1)$-juntas under the uniform distribution and cannot in general be well-approximated by juntas under relative error, but \cite{CPPS25} gave efficient testing algorithms for each of these classes.  Each of the two testers of \cite{CPPS25} was carefully specialized to the particular class at hand; it would be very interesting, though perhaps quite challenging, to give a general extension of \Cref{thm:junta-subclass-main} to classes of functions that are only uniform-distribution approximable by juntas.  A first goal in this direction is to either give an $O_{s,\eps}(1)$-query algorithm for testing the class of $s$-term DNF, or alternatively to prove an $\omega_n(1)$-query lower bound for constant $\eps$.

\subsection{Organization}
Because of space constraints we give  standard background and preliminaries in \Cref{sec:prelims}. \Cref{sec:subclass} gives our algorithm for relative-error testing of subclasses of $k$-juntas and an overview of its analysis.  The full analysis of our junta subclass tester, completing the proof of \Cref{thm:junta-subclass-main} and \Cref{cor:subclass}, is given in \Cref{sec:end-of-proof}. 
A self-contained algorithm and analysis for the simpler problem of relative-error testing of $k$-juntas, proving \Cref{thm:junta-main}, is given in \Cref{sec:juntas}.

\section{Testing Subclasses of Juntas}
\label{sec:subclass}

Recall that a class $\calC(k)$ of Boolean functions from $\zo^n$ to $\zo$ is a \emph{subclass of $k$-juntas} if (i) every $f \in \calC(k)$ is a $k$-junta, and (ii) $\calC(k)$ is closed under permuting variables (that is, if $f(x) \in \mathcal{C}(k)$ then for any permutation $\pi : [n] \to [n]$ we have that %
$f_\pi(x)=f(\pi(x))$ is also in $\mathcal{C}(k)$).
Throughout this section, $\calC(k)$ denotes a fixed and arbitrary subclass of $k$-juntas. (For example, $\calC(k)$ might be the class of decision trees over $x_1,\dots,x_n$ with at most $k$ nodes).
We further recall that $\calC(k)^* \subset  \calC(k)$ is the subset of $\calC(k)$ consisting of the functions in $\calC(k)$ that do not depend on variables $x_{k+1},\dots,x_n$. (Continuing the example from above, $\calC(k)^*$ would be the class of size-$k$ decision trees over variables $x_1,\dots,x_k$.)

\subsection{Notation}
\label{sec:some-notation}

Let $h\le k$ be a positive integer. 
Given an injective map $\sigma:[h]\rightarrow [n]$ and an assignment $y\in \{0,1\}^n$, we write $\sigma^{-1}(y)$ to denote the string
  $y\in \{0,1\}^h$
   with $y_i=x_{\sigma(i)}$ for each $i\in [h]$.
Given a $z\in \{0,1\}^h$, we write $\sigma(z)$ to denote
  the string $u\in \{0,1\}^S$, where $S=\{\sigma(i):i\in [h]\}$ is the image set of $\sigma$ and $u_{\sigma(i)}=z_i$ for each $i\in [h]$.
Given a function $g:\{0,1\}^h\rightarrow \{0,1\}$ over
  $h$ variables and an injective map $\sigma:[h]\rightarrow [n]$,
  we write $g_\sigma$ to denote the Boolean function over
  $\{0,1\}^n$ with
  $g_\sigma(x)=g(\sigma^{-1}(x))$ {(note that $g_\sigma$ is an $h$-junta and hence also a $k$-junta)}.
We will always use $f,f',$ etc.~to denote a Boolean function over $\{0,1\}^n$ and $g, g',\frakg,$ etc.~to denote a Boolean function over $\{0,1\}^h$ for some $h\le k$.
With this notation, note that %
$\reldist(f,g_\sigma)$ and $\reldist(g_\sigma,f)$ are well defined, since $f$ and $g_\sigma$ are Boolean functions on $\zo^n$. 
\def\fcj{h}

Given a positive integer $\fcj \leq k$ and $0<\kappa\le 1/2$, we now
  define a collection of functions over $\{0,1\}^h$ called 
  $\Apr(h,\kappa)$. {The idea is that $\Apr(h,\kappa)$ contains} functions over 
  $h$ variables that approximate functions in $\calC(k)$ with
  error parameter $\kappa$, in a suitable sense that we define as follows: 
\begin{flushleft}\begin{enumerate}
\item[] For each function $f\in \calC(k)^*$, let $f':\{0,1\}^n\rightarrow \{0,1\}$ be the function defined as follows:
\begin{equation}\label{eq:defg}
f'(x) := \text{arg}\max_{b \in \{0,1\}}\left\{ \Prx_{\bw \sim \{0,1\}^{{n-h}}} \big[f(x_{[h]} \circ \bw) =b\big]\right\}.
\end{equation}
Clearly $f'$ is an $h$-junta that only depends on the first $h$
  variables. 
Let $g:\{0,1\}^h\rightarrow \{0,1\}$ be the function 
  with $g(z)=f'(z\circ 0^{n-h})$.
The function $g$ is in $\Apr(h,\kappa)$ iff $\reldist(f,f')\le \kappa$. %
\end{enumerate}\end{flushleft}
 {Note that a testing algorithm can construct the set $\Apr(h,\kappa)$ ``on its own'' without making any queries.}
At a high level, the reason we need to introduce $\Apr(h,\kappa)$, instead
  of working with $\calC(k)^*$ directly, is because the algorithm may only be able to 
  identify, implicitly, an $h$-subset of the $k$ relevant variables.
 
The fact below follows directly from the definition of  $\Apr(h,\kappa)$:
\begin{fact}\label[fact]{labelfact}
$|\Apr(h,\kappa)|\le |\mathcal{C}(k)^*|$.
\end{fact} 

The next lemma shows that if  $f:\{0,1\}^n\rightarrow \{0,1\}$ is close to 
  some $g\in \Apr(h,\kappa)$ in relative distance under some
  injective map $\sigma:[h]\rightarrow [n]$, 
  then it must be close to $\mathcal{C}(k)$ in relative distance as well. 

\begin{lemma}\label[lemma]{lem:simple2}
If $f:\{0,1\}^n\rightarrow \{0,1\}$ satisfies
  $\reldist(f,g_\sigma)\le 1$ for some function $g\in \Apr(h,\kappa)$ and injective map $\sigma:[h]\rightarrow [n]$,
  then we have $\reldist(f,\mathcal{C}(k))\le \reldist(f,g_\sigma)+4\kappa$.
\end{lemma}
\begin{proof}
By the construction of $\Apr(h,\kappa)$, given that $g\in \Apr(h,\kappa)$, there exists a function $f'\in \calC(k)$ that satisfies $\reldist(f',g_\sigma)\le \kappa$.
Using \Cref{lem:approx-symetric} (in \Cref{sec:background-prop-test}) which shows that relative distance satisfies approximate symmetry, we have  
that $\reldist(g_\sigma,f')\le 2\kappa$.
The lemma follows by bounding $\reldist(f,f')$ using
  the triangle inequality (\Cref{lem:approx-triangle-ineq} in \Cref{sec:background-prop-test}).
\end{proof}

\subsection{Overview of the Main Algorithm} \label{sec:overview}

We start with an overview of our main
  algorithm, \Cref{alg:junta-subclass-tester}, for testing an arbitrary subclass $\mathcal{C}(k)$ of $k$-juntas that is closed under permutations. 
At a very high level, the approach is to perform  ``implicit learning'' by searching a candidate set of hypotheses to see whether it contains a hypothesis that is consistent with the algorithm's data.

Let $f:\{0,1\}^n\rightarrow \{0,1\}$ be the  function that is being tested.
\Cref{alg:junta-subclass-tester} starts by drawing a partition of
  variables $[n]$ into $r=O(k^2)$ many blocks $\bX_1,\ldots,\bX_r$
  uniformly at random. 
For the case when $f\in \calC(k)$, because $f$ only depends on 
  $k$ variables, a quick calculation shows that most likely 
  every block $\bX_\ell$ contains at most one relevant variable
  of $f$. For the overview we will assume that this is the case when $f\in \calC(k)$.
  
In Phase~1,  \Cref{alg:junta-subclass-tester} tries to find as many \emph{relevant} blocks $\bX_\ell$ of $f$ as possible under a query complexity budget. Here intuitively we say $\bX_\ell$ is a relevant block of $f$
  if the algorithm has found a string $v^\ell\in \{0,1\}^n$ such that 
  $f^\ell:\{0,1\}^{\bX_\ell}\rightarrow \{0,1\}$, defined as 
\begin{equation}\label{eq:temp101}
{f^\ell(x) := f\left(v^\ell_{\overline{\bX_\ell}}\circ x \right)},
\end{equation} is not a constant function.
To search for a new relevant block, 
  let $\{\smash{\bX_\ell}\}_{\smash{\ell\in \bI}}$ denote the set of relevant blocks found so far and $\bX$ be their union (with $\smash{\bI=\emptyset}$ and $\bX=\emptyset$ initially). 
Phase 1 draws~a~uniform satisfying assignment $\bu \sim f^{-1}(1)$ and a random $\bw\sim\smash{\{0,1\}^{\overline{\bX}}}$; if $f(\bu_{\bX}\circ\bw)=0$,  then a binary search over blocks can be run
  on $\bu$ and $\bu_{\bX}\circ \bw$, using $O(\log r)=O(\log k)$ many queries, to find a new relevant block $\bX^\ell$ together with an accompanying $v^\ell$ (which is used to define $f^\ell$).
If Phase 1 finds too many relevant blocks (i.e., $|\bI|$ becomes larger than $k$), then \Cref{alg:junta-subclass-tester} rejects; on the other hand, if it fails to make progress after $T_2=O(\log k)\cdot T_1$ 
many rounds, where
$$T_1=O\left(\frac{\log |\calC(k)^*|}{\eps}\right),$$ then \Cref{alg:junta-subclass-tester} stops searching and moves on to Phase 2. The latter helps ensure that the number of queries used in Phase 2 stays within the budget.

In Phase 2, \Cref{alg:junta-subclass-tester} checks if every $f^\ell$, $\ell\in \bI$, is $(1/30)$-close to a literal under the uniform distribution  (i.e., $x_{i}$ or $\overline{x_{i}}$ for some variable $i\in \bX_\ell$),
and rejects if any of them is not.

When \Cref{alg:junta-subclass-tester} reaches Phase 3, one may assume that the relevant blocks found satisfy the following two conditions (recall $\bX$ is the union of relevant blocks $\bX_\ell$, $\ell\in \bI$, found in Phase 1):
\begin{flushleft}\begin{enumerate}
\item First, we have\begin{equation}\label{temp44}
\Prx_{\substack{\bu \sim f^{-1}(1)\\ \bw \sim \{0,1\}^{\overline \bX}}}\big[f(\bu_{\bX} \circ \bw)=0\big] \le \kappa/4, %
\end{equation}
where $\kappa$ is chosen to be $1/20T_1$,
since otherwise it is unlikely for Phase 2 to fail to find a new 
  relevant block after $T_2=O(\log k)\cdot T_1$ repetitions before moving to Phase 3. 
\item Second, for every $\ell\in \bI$, $f^\ell$  is $(1/30)$-close to a literal $x_{\tau(\ell)}$ or $\overline{x_{\tau(\ell)}}$ for some $\tau(\ell)\in \bX_\ell$;
since otherwise Phase {2} would reject with high probability. When this is the case, $\tau(\ell)\in \bX_\ell$ for each $\ell\in \bI$ is well defined and unique.
\end{enumerate}
\end{flushleft}

To give some intuition for Phase 3 and 4, let's consider
  the case when $f\in \calC(k)$.
Let $h=|\bI|$ and let $\ell_1<\ldots <\ell_h$ be indices of blocks in $\bI$.
Let $\sigma:[h]\rightarrow [n]$ be the injective map 
  with $\sigma(i)=\tau(\ell_i)$ for each $i\in [h]$,
   let $S = \{\tau(\ell_i): i \in [h]\}$, and let $\sfg:\{0,1\}^h\rightarrow \{0,1\}$
  be the following function:
$$
\sfg(x):= \text{arg}\max_{b \in \{0,1\}}\left\{ \Prx_{\bw \sim \{0,1\}^{\overline{S}}} \big[f(\sigma(x)\circ\bw)=b\big]\right\}.
$$
\Cref{lem:apple} shows that \Cref{temp44} implies 
  $\reldist(f,\sfg_\sigma)\le \kappa$ and thus, we have $\sfg\in \Apr(h,\kappa)$
  by the construction of $\Apr(h,\kappa)$.
On the other hand, if $f$ is $\eps$-far from every function
  in $\mathcal{C}(k)$ in relative distance,
  then by \Cref{lem:simple2}, every function $g$ in $\Apr(h,\kappa)$ 
  has $\reldist(f,g_\sigma)>\eps/2$ given that $\kappa < \eps/8$.
The goal of Phase 3 and 4 is then to distinguish the following two cases:  (i) there exists a $\sfg\in \Apr(h,\kappa)$ such that 
  $\reldist( f,\sfg_\sigma)\le \kappa$, versus (ii) every $g\in \Apr(h,\kappa)$ has $\reldist(f,g_\sigma)\ge \eps/2$ ( note that there is a big gap between the
  two distance parameters $\kappa$ and $\eps$).  

To gain some intuition, pretend for a moment that $\Apr(h,\kappa)$ is a collection
  of functions over $\{0,1\}^n$, and the two cases we need
  to distinguish are simply either (1) there exists a function~$\sfg \in $ $\Apr(h,\kappa)$ such that $\reldist(f,\sfg)\le \kappa$ or (2)
  every $g\in \Apr(h,\kappa)$ satisfies $\reldist(f,g)$ $\ge \eps/2$.
This could be done easily as follows:
\begin{flushleft}\begin{enumerate}
\item Draw $1/(20\kappa)$  satisfying assignments $\bu\sim f^{-1}(1)$ 
  and remove every  
  $g\in \Apr(h,\kappa)$ with $g(\bu)=0$ for some sampled $\bu$.
Reject if no function is left in $\Apr(h,\kappa)$; otherwise, 
  let $\frakg$ be the function left with the smallest
  $\frakg^{-1}(1)$ and pass it to the next phase.

If we are in case (1), then with probability at least 
  $19/20$ the function $\sfg$ survives so the algorithm does not reject. 
On the other hand, in both cases (1) and (2), by a union bound every $g$ that survives
  must satisfy
$$
\frac{|f^{-1}(1)\setminus g^{-1}(1)|}{|f^{-1}(1)|}{\le \eps/500}.
$$

In particular the function $\frakg$ passed down to the next phase must satisfy this condition in both cases.
For case (1), one can use the choice of $\frakg$ (i.e., $|\frakg^{-1}(1)|$ is no larger than $|(\sfg)^{-1}(1)|$) to show that
\begin{align}\label{eq:temp100}
\frac{|\frakg^{-1}(1)\setminus f^{-1}(1)|}{|\frakg^{-1}(1)|}{\le \eps/400}.
\end{align} %
For case (2), using $\reldist(f,\frakg)\ge \eps/2$,
  one can show that the ratio above is at least $\eps/7$.%

\item Draw $20/\eps$ samples $\bu\sim {\frakg}^{-1}(1)$;
  accept if $f(\bu)=1$ for all $\bu$ sampled and reject otherwise.
Given \Cref{eq:temp100}, in case (1) we will accept with probability at least $19/20$. 
 On the other hand, in case (2), given that the ratio in \Cref{eq:temp100} is at least $\eps/7$ we will reject with probability at least
$$1-(1-\eps/7)^{20/\eps}\ge 1-e^{-20/7}\geq 0.9 .$$  %
\end{enumerate}\end{flushleft}

The discussion above for the simplified situation is 
  exactly what Phase 3 and Phase 4 
  in \Cref{alg:junta-subclass-tester} aim to achieve,
  except that in the real situation, $\Apr(h,\kappa)$ consists
  of functions over $\{0,1\}^h$ and we don't have direct access
  to the ``hidden'' injective map $\sigma$ that connects $f$ and $g$.
In Phase 3,
  \Cref{alg:junta-subclass-tester} 
  mimics step 1 above {as follows:} draw $\bu\sim f^{-1}(1)$, use
  it to \emph{obtain} $\sigma^{-1}(\bu)\in \{0,1\}^h$,
  and remove those $g$ with $g(\sigma^{-1}(\bu))=0$.
(This is because $g_\sigma(\bu)=g(\sigma^{-1}(\bu))$.)
Recall that $\frakg\in \Apr(h,\kappa)$ is the function with the smallest $|\frakg^{-1}(1)|$ that remains (recall the discussion in the final paragraph of \Cref{sec:techniques}).
  Phase 4 draws 
  $\bz\sim {\frakg}^{-1}(1)$, uses it to \emph{obtain}
  $\bu=\sigma(\bz) \circ \bw$ with $\smash{\bw\sim \{0,1\}^{\overline{S}}}$, where $S:=\{\sigma(1),\ldots,\sigma(h)\}$, and rejects if $f(\sigma(\bz)\circ\bw)=0$.
(This is because $\bu$ obtained this way is distributed 
  exactly the same as drawing uniformly from $\frakg_\sigma^{-1}(1)$.)

The challenges are (1) how to obtain $\sigma^{-1}(u)$ given $u\in \{0,1\}^n$ and (2) how to obtain $\sigma(z)\circ \bw$, $\smash{\bw\sim \{0,1\}^{\overline{S}}}$, given $z \in \{0,1\}^h$.
Both of them rely on a simple subroutine called 
   \textbf{RelVarValue}:
Given a Boolean function $\psi$ which is 
  promised to be $(1/30)$-close to an (unknown) literal under the uniform distribution and any
  assignment, \textbf{RelVarValue} can reveal 
  the value of the literal variable in that assignment with high probability. %
Given \textbf{RelVarValue}, obtaining $\sigma^{-1}(u)$ from $u\in \{0,1\}^n$ is easy: just run \textbf{RelVarValue} on $f^\ell$ and $u_{\bX_\ell}$ for each $\ell\in I$. 

Obtaining $\sigma(z)\circ \bw$ from $z$ is more involved.
Given a $y\in \{0,1\}^n$ and a $z \in \zo^h$, we write $y_{\sigma,z}$ to denote the
  following $n$-bit string:
For each block $\bX_\ell$ with $\ell\notin \bI$, $y_{\sigma,z}$ agrees with $y$ in $\bX_\ell$; for each $\ell_i\in \bI$,
$$
\text{the $\bX_{\ell_i}$ block of $y_{\sigma,z}$}= \begin{cases} y_{\bX_{\ell_i}}& \text{if $y_{\tau(\ell_i)}=z_i$}\\[0.5ex]\overline{y_{\bX_{\ell_i}}} & \text{if $y_{\tau(\ell_i)}\ne z_i$}
\end{cases}$$
A key observation is that, when $\by\sim \{0,1\}^n$, $\by_{\sigma,z}$ is distributed exactly the same as $\sigma (z)\circ \bw$
  with $\smash{\bw\sim\{0,1\}^{\overline{S}}}$.
The $\MapBack$ subroutine, which calls on $\RelV$, takes 
  $y$ and $z$ as inputs and outputs $y_{\sigma,z}$ with high probability.

The rest of the argument for testing subclasses of juntas %
is organized as follows.
We first focus on the two subroutines $\RelV$ and $\MapBack$ and 
  analyze their performance guarantees in \Cref{sec:RelVarValue}.
In \Cref{subsec:yescase} we show that when $f\in \calC(k)$, \Cref{alg:junta-subclass-tester} accepts with probability
  at least $2/3$;
  in \Cref{subsec:nocase}, we show that when $f$ is $\eps$-far from $\calC(k)$ in relative 
  distance, \Cref{alg:junta-subclass-tester} rejects with probability at least $2/3$.
Finally we bound the query complexity of \Cref{alg:junta-subclass-tester} in \Cref{subsec:querycomp}.

\begin{algorithm}

\caption{{\textsc{Junta-Subclass-Tester}}$(f,\eps)$}
\label{alg:junta-subclass-tester}

\Algphase{Phase 1: Partition $[n]$ and find relevant blocks}
 
\nonl Parameters used in the algorithm:
$$
\hspace{-1.3cm}r=20k^2,\quad T_1=O\left(\frac{\log |\calC(k)^* |}{\eps}\right),\quad T_2= 100\log (20k+1)\cdot T_1\quad\text{and}\quad \kappa=\frac{1}{20T_1}.
$$

Randomly partition $[n]$ into $r$ blocks $\bX_1, \ldots, \bX_r$; Set $\bX=\emptyset$, $\bI=\emptyset$ and $t=0$\; %
\While{$t \neq T_2$}{
 \hskip0em {Draw $\bu \sim f^{-1}(1)$ and $\smash{\bw \sim \{0,1\}^{\overline{\bX}}}$}; set $t\leftarrow t+1$\;
 
 \hskip0em \If{$f(\bu_{\bX} \circ \bw) \neq f(\bu)$}{ 

Binary Search over blocks of $\bX_1,\dots,\bX_r$ that are subsets of $\overline{\bX}$ to find 
\hskip0em  a new relevant block $\bX_\ell$ and a string $v^{\ell} \in \{0,1\}^n$ such that $f(v^{\ell}) \neq f\left(v^{\ell}_{\overline{\bX_\ell}}\circ \bw_{\bX_\ell} \right).$ \\
\nonl
Denote by $f^{\ell} : \{0,1\}^{\bX_\ell} \rightarrow \{0,1\}$ the following function: $$f^{\ell}(x):=f\left(v^{\ell}_{\overline{\bX_\ell}} \circ x  \right).$$

Set $\bX \leftarrow \bX \cup \bX_\ell$, $\bI \leftarrow \bI \cup \{\ell\}$, %
 $t\leftarrow 0$, and 
 \text{reject} if $|\bI|>k$\;

 }
}

\Algphase{Phase 2: Test the relevant sets} %

\ForEach{$\ell \in \bI$}{
 \hskip0em Run {$\textsc{UniformJunta}(f^{\ell}, 1, 1/15, 1/30)$}; \text{reject} if it rejects\; %
 \hskip0em Draw $\bb \sim \{0,1\}^{{\bX}_\ell}$; \text{Reject} if ${f^{\ell}(\bb)= f^\ell (\overline{\bb} )}$\;
}

\Algphase{Phase 3:  Trimming $\Apr$}

 Set $\bA\leftarrow \Apr(h, \kappa)$, where $h=|\bI|$ and  $\bI=\{\ell_1,\ldots,\ell_h\}$
  with $\ell_1<\ldots<\ell_h$\;
\RepTimes{$T_1$}{
Draw $\bu\sim f^{-1}(1)$ and let $\bv$ be the all-$0$ string in $\{0,1\}^h$\;
\For{$i=1$ \KwTo$h$}{
 Set $\smash{\bv_i\leftarrow \smash{\RelV(f^{\ell_i},\bu_{\bX_{\ell_i}},1/(20k))}}$\;}

 \hskip0em Remove from $\bA$ every $g\in \bA$ with  $g(\bv)=0$\;
 }
 \text{reject} if $\bA=\emptyset$; otherwise let $\frakg$ be the function in $\bA$ with the smallest $|\frakg^{-1}(1)|$;

\Algphase{Phase 4:  Find $k$-variable approximator for $f$} %

\RepTimes{$20/\eps$}{
Draw $\by\sim \{0,1\}^n$ and $\bz\sim {\frakg}^{-1}(1)$ (note that $\bz \in \{0,1\}^h$)\;
Let $\bu= \MapBack(f,\bI,\{\bX_\ell:\ell\in \bI\},\{v^\ell:\ell\in \bI\},\by,\bz)$\;
  \text{reject} if %
  $f(\bu)=0$\;
  }
 \text{Accept}\vspace{0.1cm}\;
\end{algorithm}

\begin{algorithm}
\caption{$\RelV(\psi, w,\delta)$}

  Let $Y_{\zeta}=\{j: w_j = \zeta\}$ for both $\zeta\in \zo$\;
Set $\bG_{0}=
\bG_{1}=0$\;
\RepTimes{$O(\log(1/\delta))$}{
 \hskip0em Draw $\mathbf{b^\zeta} \sim \{0,1\}^{Y_{\zeta}}$ for both $\zeta \in \zo$\;
 \hskip0em If $\smash{{\psi(\mathbf{b}^0 \circ \bb^1 ) \neq \psi(\overline{\mathbf{b}^0} \circ \bb^1)}}$ then set $\bG_{0} \leftarrow \bG_{0}+1$; otherwise,  set $\bG_{1} \leftarrow \bG_{1}+1$\;
 }
  Return $0$ if $\bG_0>\bG_1$ and return $1$ otherwise\;

\end{algorithm}

\subsection{The $\RelV$ and $\MapBack$ subroutines}\label{sec:RelVarValue}%

The $\RelV$ subroutine takes three inputs: a Boolean function $\psi:\{0,1\}^m\rightarrow \{0,1\}$, an assignment
  $w\in \{0,1\}^m$, and an error parameter $\delta>0$.
When $\psi$ is promised to be $(1/30)$-close to a literal (say $x_i$ or $\overline{x_i}$) under the uniform 
  distribution, $\RelV$ 
  returns the value $w_i$ of the literal variable in $w$ with probability at least $1-\delta$.

\begin{lemma}\label[lemma]{lem: FindBlockValue: bad}
$\RelV(\psi,w,\delta)$ makes $O(\log(1/\delta))$ queries and satisfies the following conditions: 
\begin{flushleft}\begin{enumerate}
  \item If $\psi : \{0,1\}^m \to \{0,1\}$ is $(1/30)$-close to a literal $x_i$ or $\overline{x_i}$ under the uniform distribution, then with probability at least $1-\delta$, $\RelV$ returns the bit $w_i$.
\item If $\psi$ is a literal $x_i$ or $\overline{x_i}$, then $\RelV$ always returns $w_i$.
\end{enumerate}\end{flushleft}
\end{lemma}
\begin{proof}
The case when $\psi$ is a literal is simple as
 the correct counter always gets incremented.

Consider the case when $\psi$ is $(1/30)$-close
  to either $x_i$ or $\overline{x_i}$ with $w_i=0$ (meaning $i\in Y_0$).
Then
$$
\Pr\big[\psi(\bb^0\circ\bb^1)= \psi(\overline{\bb^0}\circ\bb^1)\big]
\le \Pr\big[\psi(\bb^0\circ\bb^1)\ne \bb^0_i\big]
+\Pr\big[\psi(\overline{\bb^0}\circ\bb^1)\ne \overline{\bb^0_i}\big]\le 1/15,
$$
where the last inequality follows from the fact that 
  both the marginal distributions of $\bb^0\circ \bb^1$ and 
  $\smash{\overline{\bb^0}\circ\bb^1}$ are uniform.
As a result, in each round of $\RelV$, $\bG_0$ goes up by
  $1$ with probability at least $14/15$ and $\bG_1$ goes
  up by $1$ with probability at most $1/15$.
The claim then follows from a standard Chernoff bound in this case, by making the hidden constant large enough. 
  The case when $w_i=1$ and $i\in Y_1$ can be proved similarly, using 
$$
\Pr\big[\psi(\bb^0\circ\bb^1)\ne \psi(\overline{\bb^0}\circ\bb^1)\big]
\le \Pr\big[\psi(\bb^0\circ\bb^1)\ne \bb^1_i\big]
+\Pr\big[\psi(\overline{\bb^0}\circ\bb^1)\ne {\bb^1_i}\big]\le 1/15.
$$
This finishes the proof of the lemma.%
\end{proof}

\begin{algorithm}[t]
\caption{$\MapBack(f,I,\{X_\ell:\ell\in I\},\{v^\ell:\ell\in I\},y,z)$} %
\label{alg:SAMPfA}

\ForEach{$i\in [h]$}{
 \hskip0em Let $\smash{b_{i}= \RelV(f^{\ell_i},y_{X_{\ell_i}},1/(20k))}$\;
 \hskip0em Flip the $X_{\ell_i}$-block of $y$ to $\overline{y_{X_{\ell_i}}}$ if $b_i\ne z_i$\;
 }
 Return $y$\;
\end{algorithm}

Next we work on the subroutine $\MapBack$ used in 
  Phase 4 of \Cref{alg:junta-subclass-tester},
  which makes calls to $\RelV$.
It takes as input the function $f$; the set $I$; for each $\ell\in I$ a collection of $h\le k$ blocks
  $X_\ell$ together with a string $v^\ell\in \{0,1\}^n$
  such that the function $f^\ell$ over $\{0,1\}^{X_\ell}$ defined using 
  $v^\ell$ in \Cref{eq:temp101} is $(1/30)$-close to 
  a literal $x_{\tau(\ell)}$ or $\overline{x_{\tau(\ell)}}$
  under the uniform distribution; and two strings $y\in \{0,1\}^n$ and  $z\in \{0,1\}^h$.
Let $\ell_1<\cdots<\ell_h$ be the indices in $I$, and let
  $\sigma$ and $S$ be defined using $\tau$ {as in \Cref{sec:overview}.}

\begin{lemma}\label[lemma]{lem:mapback}
$\MapBack$ uses $O(k\log k)$ queries and  satisfies the following conditions: 
\begin{flushleft}\begin{enumerate}
\item If $f^\ell$ is $(1/30)$-close to a literal $x_{\tau(\ell)}$ or $\overline{x_{\tau(\ell)}}$ for all $\ell\in I$ under the uniform distribution, then with probability at least $1-1/20$, $\MapBack$  returns 
  $y_{\sigma,z}$.

\item If $f^\ell$ is a literal for all $\ell\in I$, then
  $\MapBack$ always returns $y_{\sigma,z}$.
\end{enumerate}\end{flushleft}
\end{lemma}

\begin{proof}
The query complexity follows from the query complexity of $\RelV$, and 
  our choice of $\delta=1/(20k)$.
The probability follows from a union bound over the $h\le k$ calls to $\RelV$.
\end{proof} %
%
%

%

%

%
%

%

%
%

%

%
%
%

%
%
%

%
%
%
%
%
%
%
%
%
%
%
%
%

%

%

%
%
%
%
%

%

%
%
    
%
%
%
%
%
%
%
%

%
%
%
%
%
%
%
%
%

%
    
%
%
%
%
%
%

%

%
%
%
%

%
%
%

%
%
    
%
%
%
%
%
%
%

%
%

%
%
%
%
%
%

%
%
%
%
%
%
%
%
%
%
%
%
%
%

%
%
%
%
%
%
%
%
%
%
%
%
%

%

%
%
%
%
%

%

%
%
%
%
%

%

%

%

%
%

%

%
%
%
%
%
%
%

%
%

%

%

%
%

%
%
%

%
%
%

%
%
%

%

%

%
%
%

%
%
%
%

 %
%
%
%
%

%
%

%
%
%
%
%

%

%
%
%
%
%
%
%
%
%
%
%
%
%

%
%
%

%

%
%
%
%
%
%

%
%

%

%
    
%
%

%

%
%
%
%
%
%
%
%
%
%

%

%
%
%
%
%
%
%

%
%

%
%
%

%
%
%
%
%
%
%
%
%
%
%

%

%

%
%
%
%
%
%
%
%

%

%

%
%

%
%

%

%

%
%
%

%
%

%
%
%

%
%
%
%
%

%
%
%
%
%
%
%
%
%
%

%
%

%
%
%
%
%

%

%
%


\bibliographystyle{alpha}
\bibliography{allrefs}

\newcommand{\etalchar}[1]{$^{#1}$}
\begin{thebibliography}{BGMdW13}

\bibitem[ABRdW16]{ABRW16}
Andris Ambainis, Aleksandrs Belovs, Oded Regev, and Ronald de~Wolf.
\newblock Efficient quantum algorithms for (gapped) group testing and junta testing.
\newblock In {\em Proceedings of the Twenty-Seventh Annual {ACM-SIAM} Symposium on Discrete Algorithms, {SODA} 2016}, pages 903--922, 2016.

\bibitem[Alo12]{Alon12}
Noga Alon.
\newblock Lower bound on active testing dimension of juntas.
\newblock Personal communication to authors of \cite{BBBY12} (see p. 3 of \cite{BBBY12}), 2012.

\bibitem[AS07]{AticiServedio:07qip}
A.~At{\i}c{\i} and R.~Servedio.
\newblock Quantum algorithms for testing and learning juntas.
\newblock {\em Quantum Information Processing}, 6(5):323--348, 2007.

\bibitem[BBBY12]{BBBY12}
Maria{-}Florina Balcan, Eric Blais, Avrim Blum, and Liu Yang.
\newblock {Active Property Testing}.
\newblock In {\em 53rd Annual {IEEE} Symposium on Foundations of Computer Science, {FOCS}}, pages 21--30, 2012.

\bibitem[Bel19]{belovs2019quantum}
Aleksandrs Belovs.
\newblock Quantum algorithm for distribution-free junta testing.
\newblock In {\em Computer Science--Theory and Applications: 14th International Computer Science Symposium in Russia, CSR 2019}, pages 50--59. Springer, 2019.

\bibitem[BGMdW13]{BGMW13}
Harry Buhrman, David Garc{\'{\i}}a{-}Soriano, Arie Matsliah, and Ronald de~Wolf.
\newblock The non-adaptive query complexity of testing k-parities.
\newblock {\em Chic. J. Theor. Comput. Sci.}, 2013, 2013.

\bibitem[Bla08]{Blais08}
Eric Blais.
\newblock Improved bounds for testing juntas.
\newblock In {\em Proc. RANDOM}, pages 317--330, 2008.

\bibitem[Bla09]{Blaisstoc09}
Eric Blais.
\newblock Testing juntas nearly optimally.
\newblock In {\em Proc.\ 41st Annual ACM Symposium on Theory of Computing (STOC)}, pages 151--158, 2009.

\bibitem[BLR93]{BLR93}
Manuel Blum, Michael Luby, and Ronitt Rubinfeld.
\newblock Self-testing/correcting with applications to numerical problems.
\newblock {\em Journal of Computer and System Sciences}, 47:549--595, 1993.
\newblock Earlier version in STOC'90.

\bibitem[Blu94]{Blum:94}
A.~Blum.
\newblock Relevant examples and relevant features: Thoughts from computational learning theory.
\newblock in AAAI Fall Symposium on `Relevance', 1994.

\bibitem[Blu03]{blu03}
Avrim Blum.
\newblock Learning a function of $r$ relevant variables.
\newblock In Bernhard Sch{\"o}lkopf and Manfred~K. Warmuth, editors, {\em Proc.\ 16th Annual Conference on Learning Theory (COLT)}, volume 2777 of {\em Lecture Notes in Computer Science}, pages 731--733. Springer-Verlag, 2003.

\bibitem[Bou02]{Bourgain}
J.~Bourgain.
\newblock On the distributions of the {F}ourier spectrum of {B}oolean functions.
\newblock {\em Israel J. Math.}, 131:269--276, 2002.

\bibitem[Bsh19]{Bshouty19}
Nader~H. Bshouty.
\newblock Almost optimal distribution-free junta testing.
\newblock In {\em 34th Computational Complexity Conference, {CCC}}, pages 2:1--2:13, 2019.

\bibitem[Bsh20]{Bshouty20}
Nader~H. Bshouty.
\newblock Almost optimal testers for concise representations.
\newblock In {\em Approximation, Randomization, and Combinatorial Optimization. Algorithms and Techniques, {APPROX/RANDOM}}, pages 5:1--5:20, 2020.

\bibitem[CDH{\etalchar{+}}25a]{CDHLNSY2024}
X.~Chen, A.~De, Y.~Huang, S.~Nadimpalli, R.~Servedio, and T.~Yang.
\newblock {Relative error monotonicity testing}.
\newblock In {\em {Proc. ACM-SIAM Symposium on Discrete Algorithms (SODA)}}, 2025.

\bibitem[CDH{\etalchar{+}}25b]{CDHNSY25}
X.~Chen, A.~De, Y.~Huang, S.~Nadimpalli, R.~A. Servedio, and T.~Yang.
\newblock {Testing halfspaces with relative error}.
\newblock Manuscript, 2025.

\bibitem[CDL{\etalchar{+}}24]{CDLNS24}
Xi~Chen, Anindya De, Yuhao Li, Shivam Nadimpalli, and Rocco~A. Servedio.
\newblock Mildly exponential lower bounds on tolerant testers for monotonicity, unateness, and juntas.
\newblock In {\em Proceedings of the 2024 {ACM-SIAM} Symposium on Discrete Algorithms, {SODA} 2024}, pages 4321--4337, 2024.

\bibitem[CFP24]{CFP24}
Xi~Chen, Yumou Fei, and Shyamal Patel.
\newblock Distribution-free testing of decision lists with a sublinear number of queries.
\newblock In {\em Proceedings of the 56th Annual {ACM} Symposium on Theory of Computing (STOC)}, pages 1051--1062, 2024.

\bibitem[CG04]{ChocklerGutfreund:04}
H.~Chockler and D.~Gutfreund.
\newblock A lower bound for testing juntas.
\newblock {\em Information Processing Letters}, 90(6):301--305, 2004.

\bibitem[CGM11]{CGM11}
Sourav Chakraborty, David Garc{\'{\i}}a{-}Soriano, and Arie Matsliah.
\newblock Efficient sample extractors for juntas with applications.
\newblock In {\em Automata, Languages and Programming - 38th International Colloquium, {ICALP}}, pages 545--556, 2011.

\bibitem[CNY23]{CNY23}
Thomas Chen, Shivam Nadimpalli, and Henry Yuen.
\newblock Testing and learning quantum juntas nearly optimally.
\newblock In {\em Proceedings of the 2023 {ACM-SIAM} Symposium on Discrete Algorithms, {SODA} 2023}, pages 1163--1185, 2023.

\bibitem[CP22]{CP22}
Xi~Chen and Shyamal Patel.
\newblock Distribution-free testing for halfspaces (almost) requires {PAC} learning.
\newblock In {\em Proceedings of the 2022 {ACM-SIAM} Symposium on Discrete Algorithms (SODA)}, pages 1715--1743, 2022.

\bibitem[CP23]{ChenPatel23}
Xi~Chen and Shyamal Patel.
\newblock New lower bounds for adaptive tolerant junta testing.
\newblock {\em FOCS}, pages 1778--1786, 2023.

\bibitem[CPPS25]{CPPS25}
X.~Chen, W.~Pires, T.~Pitassi, and R.~A. Servedio.
\newblock {Relative-error testing of conjunctions and decision lists}.
\newblock Manuscript, 2025.

\bibitem[CST{\etalchar{+}}18]{CSTWX18jacm}
Xi~Chen, Rocco~A. Servedio, Li{-}Yang Tan, Erik Waingarten, and Jinyu Xie.
\newblock Settling the query complexity of non-adaptive junta testing.
\newblock {\em J. {ACM}}, 65(6):40:1--40:18, 2018.

\bibitem[CWX17]{CWX17stoc}
Xi~Chen, Erik Waingarten, and Jinyu Xie.
\newblock {Beyond Talagrand functions: new lower bounds for testing monotonicity and unateness}.
\newblock In {\em Proceedings of the 49th Annual {ACM} {SIGACT} Symposium on Theory of Computing (STOC)}, pages 523--536, 2017.

\bibitem[CX16]{CX16}
Xi~Chen and Jinyu Xie.
\newblock Tight bounds for the distribution-free testing of monotone conjunctions.
\newblock In {\em Proceedings of the Twenty-Seventh Annual {ACM-SIAM} Symposium on Discrete Algorithms (SODA)}, pages 54--71, 2016.

\bibitem[DLM{\etalchar{+}}07]{DLM+:07}
I.~Diakonikolas, H.~Lee, K.~Matulef, K.~Onak, R.~Rubinfeld, R.~Servedio, and A.~Wan.
\newblock Testing for concise representations.
\newblock In {\em Proc. 48th Ann. Symposium on Computer Science (FOCS)}, pages 549--558, 2007.

\bibitem[DMN19]{DMN19}
Anindya De, Elchanan Mossel, and Joe Neeman.
\newblock Junta correlation is testable.
\newblock In David Zuckerman, editor, {\em 60th {IEEE} Annual Symposium on Foundations of Computer Science, {FOCS} 2019, Baltimore, Maryland, USA, November 9-12, 2019}, pages 1549--1563. {IEEE} Computer Society, 2019.

\bibitem[DR11]{DolevRon11}
Elya Dolev and Dana Ron.
\newblock Distribution-free testing for monomials with a sublinear number of queries.
\newblock {\em Theory of Computing}, 7(11):155--176, 2011.

\bibitem[Fil04]{Filmus22}
Y.~Filmus.
\newblock {A simple proof of the Kindler–Safra theorem}.
\newblock Manuscript, available at \href{https://yuvalfilmus.cs.technion.ac.il/Manuscripts/KindlerSafra.pdf}{https://yuvalfilmus.cs.technion.ac.il/Manuscripts/KindlerSafra.pdf}, 2004.

\bibitem[FKR{\etalchar{+}}04]{FKRSS03}
E.~Fischer, G.~Kindler, D.~Ron, S.~Safra, and A.~Samorodnitsky.
\newblock Testing juntas.
\newblock {\em J. Computer \& System Sciences}, 68(4):753--787, 2004.

\bibitem[Fri98]{Friedgut:98}
E.~Friedgut.
\newblock Boolean functions with low average sensitivity depend on few coordinates.
\newblock {\em Combinatorica}, 18(1):474--483, 1998.

\bibitem[GGR98]{GGR98}
Oded Goldreich, Shafi Goldwasser, and Dana Ron.
\newblock Property testing and its connection to learning and approximation.
\newblock {\em Journal of the ACM}, 45:653--750, 1998.

\bibitem[Gol17]{Goldreich17book}
O.~Goldreich.
\newblock {\em {Introduction to Property Testing}}.
\newblock Cambridge University Press, 2017.

\bibitem[GR02]{GoldreichRon02}
Oded Goldreich and Dana Ron.
\newblock Property testing in bounded degree graphs.
\newblock {\em Algorithmica}, 32(2):302--343, 2002.

\bibitem[GS09]{GlasnerServedio:09toc}
Dana Glasner and Rocco~A. Servedio.
\newblock Distribution-free testing lower bound for basic boolean functions.
\newblock {\em Theory of Computing}, 5(1):191--216, 2009.

\bibitem[HK04]{HalevyKushilevitz:08}
S.~Halevy and E.~Kushilevitz.
\newblock {Distribution-Free Connectivity Testing for Sparse Graphs}.
\newblock {\em Algorithmica}, 51(1):24--48, 2004.

\bibitem[HK07]{HalevyKushilevitz:07}
S.~Halevy and E.~Kushilevitz.
\newblock {Distribution-Free Property Testing}.
\newblock {\em SIAM J. Comput.}, 37(4):1107--1138, 2007.

\bibitem[ITW21]{ITW21}
Vishnu Iyer, Avishay Tal, and Michael Whitmeyer.
\newblock Junta distance approximation with sub-exponential queries.
\newblock In Valentine Kabanets, editor, {\em 36th Computational Complexity Conference, {CCC} 2021, July 20-23, 2021, Toronto, Ontario, Canada (Virtual Conference)}, volume 200 of {\em LIPIcs}, pages 24:1--24:38. Schloss Dagstuhl - Leibniz-Zentrum f{\"{u}}r Informatik, 2021.

\bibitem[KMS18]{KMS18}
Subhash Khot, Dor Minzer, and Muli Safra.
\newblock On monotonicity testing and boolean isoperimetric-type theorems.
\newblock {\em {SIAM} J. Comput.}, 47(6):2238--2276, 2018.

\bibitem[KS04]{KindlerSafra04}
G.~Kindler and S.~Safra.
\newblock Noise resistant boolean functions are juntas.
\newblock Manuscript, 2004.

\bibitem[LCS{\etalchar{+}}19]{LCSSX19}
Zhengyang Liu, Xi~Chen, Rocco~A. Servedio, Ying Sheng, and Jinyu Xie.
\newblock Distribution-free junta testing.
\newblock {\em {ACM} Trans. Algorithms}, 15(1):1:1--1:23, 2019.

\bibitem[LW19]{LW19}
Amit Levi and Erik Waingarten.
\newblock Lower bounds for tolerant junta and unateness testing via rejection sampling of graphs.
\newblock In Avrim Blum, editor, {\em 10th Innovations in Theoretical Computer Science Conference, {ITCS} 2019, January 10-12, 2019, San Diego, California, {USA}}, volume 124 of {\em LIPIcs}, pages 52:1--52:20. Schloss Dagstuhl - Leibniz-Zentrum f{\"{u}}r Informatik, 2019.

\bibitem[MORS10a]{MORS:10}
K.~Matulef, R.~O'Donnell, R.~Rubinfeld, and R.~Servedio.
\newblock Testing halfspaces.
\newblock {\em SIAM J. on Comput.}, 39(5):2004--2047, 2010.

\bibitem[MORS10b]{MORS10}
Kevin Matulef, Ryan O'Donnell, Ronitt Rubinfeld, and Rocco~A. Servedio.
\newblock Testing halfspaces.
\newblock {\em SIAM Journal on Computing}, 39(5):2004--2047, 2010.

\bibitem[MOS03]{MOS:03}
E.~Mossell, R.~O'Donnell, and R.~Servedio.
\newblock Learning juntas.
\newblock {\em Proceedings of the Thirty-Fifth Annual Symposium on Theory of Computing}, 2003.

\bibitem[NP24]{NadimpalliPatel24}
Shivam Nadimpalli and Shyamal Patel.
\newblock Optimal non-adaptive tolerant junta testing via local estimators.
\newblock In {\em Proceedings of the 56th Annual {ACM} Symposium on Theory of Computing, {STOC} 2024}, pages 1039--1050, 2024.

\bibitem[NS92]{NisanSzegedy:92}
N.~Nisan and M.~Szegedy.
\newblock On the degree of {B}oolean functions as real polynomials.
\newblock In {\em Proc. Twenty-Fourth Annual Symposium on Theory of Computing}, pages 462--467, 1992.

\bibitem[PR02]{ParnasRon02}
Michal Parnas and Dana Ron.
\newblock Testing the diameter of graphs.
\newblock {\em Random Struct. Algorithms}, 20(2):165--183, 2002.

\bibitem[PRS02]{PRS02}
M.~Parnas, D.~Ron, and A.~Samorodnitsky.
\newblock {Testing Basic Boolean Formulae}.
\newblock {\em SIAM J. Disc. Math.}, 16:20--46, 2002.

\bibitem[PRW22]{PRW22}
Ramesh Krishnan~S. Pallavoor, Sofya Raskhodnikova, and Erik Waingarten.
\newblock Approximating the distance to monotonicity of boolean functions.
\newblock {\em Random Struct. Algorithms}, 60(2):233--260, 2022.

\bibitem[RS96]{RS96}
R.~Rubinfeld and M.~Sudan.
\newblock Robust characterizations of polynomials with applications to program testing.
\newblock {\em SIAM Journal on Computing}, 25:252--271, 1996.

\bibitem[Sa{\u{g}}18]{Saglam18}
Mert Sa{\u{g}}lam.
\newblock Near log-convexity of measured heat in (discrete) time and consequences.
\newblock In Mikkel Thorup, editor, {\em 59th {IEEE} Annual Symposium on Foundations of Computer Science, {FOCS} 2018, Paris, France, October 7-9, 2018}, pages 967--978. {IEEE} Computer Society, 2018.

\bibitem[STW15]{STW15}
R.A. Servedio, L.-Y. Tan, and J.~Wright.
\newblock Adaptivity helps for testing juntas.
\newblock In {\em Proceedings of the 30th IEEE Conference on Computational Complexity}, pages 264--279, 2015.

\bibitem[Val15]{Valiant15}
Gregory Valiant.
\newblock Finding correlations in subquadratic time, with applications to learning parities and the closest pair problem.
\newblock {\em J. {ACM}}, 62(2):13:1--13:45, 2015.

\end{thebibliography}
\clearpage
\appendix

\section{Background and Preliminaries} \label{sec:prelims}

\subsection{Basic notation and terminology}

We use $[n]$ to denote the set $\{1, \ldots, n\}$. Given a set $S \subseteq [n]$, we denote by $\overline{S}$ the set $[n] \setminus S$. 

Given $S \subseteq [n]$ and $z \in \{0,1\}^n$, we denote by $z_S$ the string in $\{0,1\}^S$ that agrees with $z$ on every coordinate in $S$.
Given $z \in \zo^[n]$ we denote by $\overline{z}$ the string with each bit of $z$ flipped, so $\overline{z}_i=1-z_i$ for every $i \in [n]$. Given strings $x,y \in \{0,1\}^n$, we denote by $x \oplus y$ the bit-wise xor of $x$ and $y$, that is $(x \oplus y)_i=x_i \oplus y_i$.
Given a permutation $\pi:[n] \to [n]$ and $x \in \zo^n$, we denote by $\pi(x)$ the string $x_{\pi(1)} x_{\pi(2)} \ldots x_{\pi(n)}$. Given $f:\zo^n \to \zo$, we denote by $f_{\pi}$ the function $f_{\pi}(x)=f(\pi(x))$.

For $J \subseteq [n]$ we say that a function $f: \zo^n \to \zo$ is a \emph{$J$-junta} if $f$ depends only on the variables in $J$, i.e.~$f(x_J \circ z_{\overline{J}})=f(x_J \circ z'_{\overline{J}})$ for every $x,z,z' \in \zo^n$.  If $f$ is a $J$-junta for some set $J$ of size $k$, we say that $f$ is a \emph{$k$-junta}.

\subsection{Standard-model property testing and relative-error property testing} 
\label{sec:background-prop-test}

We briefly review the ``standard'' model for testing Boolean functions \cite{Goldreich17book}. A standard-model testing algorithm for a class ${\cal C}$ of $n$-variable Boolean functions has oracle access $\MQ(f)$ to an unknown and arbitrary function $f: \zo^n \to \zo$.
The requirement on the algorithm is that it must output ``yes'' with high probability (say at least 2/3; this can be amplified using standard techniques) if $f \in {\cal C}$, and must output ``no'' with high probability (again, say at least 2/3) if $f$ disagrees with every function $g \in {\cal C}$ on at least $\eps 2^n$ inputs in $\zo^n$.

As defined and studied in \cite{CDHLNSY2024,CDHNSY25,CPPS25}, a \emph{relative-error} testing algorithm for ${\cal C}$ has oracle access to $\MQ(f)$ and also has access to a ``random satisfying assignment'' oracle which, when called, returns a uniform random element $\bx \sim f^{-1}(1)$.
An $\eps$-relative-error testing algorithm for ${\cal C}$ must output ``yes'' with high probability (say at least 2/3; again this can be easily amplified) if $f \in {\cal C}$ and must output ``no'' with high probability (again, say at least 2/3) if $\reldist(f,{\cal C}) \geq \eps$, where
$$\reldist(f,{\cal C}):=\min_{g \in {\cal C}}\hspace{0.05cm}\reldist(f,g)\ \quad\text{and}\ \quad 
\reldist(f,g) := {\frac {|f^{-1}(1) \ \triangle \ g^{-1}(1)|}{|f^{-1}(1)|}}.
$$
Thus, in the relative-error setting, the distance from $f$ to ${\cal C}$ is measured ``relative to the sparsity of $f$.''

The following easy observation says that relative distance doesn't change when we apply a permutation:
\begin{observation}\label[observation]{obs:permutation_dist}
    If $\pi:[n] \to [n]$ is a permutation, then $\reldist(f,g)=\reldist(f_\pi, g_\pi)$. 
\end{observation}

We will need the following two easy lemmas about relative distance:

\begin{restatable}[Approximate symmetry of relative distance]{lemma}{approxsymmetry}\label[lemma]{lem:approx-symetric}
    Let $f,g:\{0,1\}^n \to \{0,1\}$ be functions  such that $\reldist(f,g)\leq \epsilon$ where $\epsilon \leq 1/2$. Then $\reldist(g,f) \leq 2\epsilon$.
\end{restatable} 
\begin{proof}
Since $\reldist(f,g)\leq \epsilon$ we have that 
$$ \frac{|f^{-1}(1) \triangle g^{-1}(1)|}{|f^{-1}(1)|} \leq \epsilon \leq \frac{1}{2}.$$
which means we must have that $|g^{-1}(1)| \geq \frac{|f^{-1}(1)|}{2}$. Hence $$\reldist(g,f)=\frac{|f^{-1}(1) \triangle g^{-1}(1)|}{|g^{-1}(1)|}\leq \frac{|f^{-1}(1) \triangle g^{-1}(1)|}{2|f^{-1}(1)|} \leq  2\reldist(f,g) \leq 2\epsilon. $$
\end{proof}

\begin{restatable}[Approximate triangle inequality for relative distance]{lemma}{approxtriangleinequality}\label[lemma]{lem:approx-triangle-ineq}
    Let $f,g,h:\{0,1\}^n \to \{0,1\}$ be such that $\reldist(f,g)\leq \epsilon$ and $\reldist(g,h)\leq \epsilon'$. Then $\reldist(f,h) \leq \epsilon+(1+\epsilon)\epsilon'$.
\end{restatable}

\begin{proof}
    Since $\reldist(f,g)\leq \epsilon$ we have that $${\frac {|f^{-1}(1) \ \triangle \ g^{-1}(1)|}{|f^{-1}(1)|}} \leq \eps,$$ which means we must have $|g^{-1}(1)| \leq (1+\epsilon)|f^{-1}(1)|$. Hence
    \begin{align*}
        \reldist(f,h) &= {\frac {|f^{-1}(1) \ \triangle \ h^{-1}(1)|}{|f^{-1}(1)|}}\\
        &=\frac{1}{|f^{-1}(1)|}\left| \{ z \::\: f(z) \neq h(z) \}\right| \\
        & \leq \frac{1}{|f^{-1}(1)|}\left| \{ z \::\: f(z) \neq g(z) \lor g(z) \neq h(z) \}\right| \\
        & \leq \frac{1}{|f^{-1}(1)|}\left| \{ z \::\: f(z) \neq g(z) \}\right| + \frac{1}{|f^{-1}(1)|} \left| \{z :  g(z) \neq h(z) \}\right| \\
        & \leq  \reldist(f,g) + (1+\epsilon) \frac{1}{|g^{-1}(1)|} \left| \{z : g(z) \neq h(z) \}\right| \\
        &= \reldist(f,g) + (1+\epsilon) \reldist(g,h) \\
        &\leq \epsilon+(1+{\epsilon})\epsilon'. 
    \end{align*}
\end{proof}

We will use the following one-sided algorithm for junta testing in the standard setting \cite{Blaisstoc09}:

\begin{theorem}\label{thm: uniform junta blais}
    There exists a one-sided adaptive algorithm, $\textsc{UniformJunta}(f,k,\epsilon , \delta)$, for $\epsilon$-testing $k$-juntas that
makes $O((k/\epsilon +k\log k) \log(1/\delta))$ queries. The algorithm always accepts $f$ if it is a $k$-junta
  and rejects $f$ with probability at least $1-\delta$ if it is $\epsilon$-far from
every $k$-junta with respect to the uniform distribution.
\end{theorem}

A key subroutine our algorithm will use is binary search over ``blocks'' of variables to find new relevant blocks. In more detail, let $X_{1}, \ldots X_{r}$ be a partition of $[n]$ (we will frequently refer to a set $X_i$ as a ``block''). Fix some $I \subseteq [r]$, and let $X=\cup_{\ell \in I} X_\ell$. Given $u \in \zo^n$ and $w \in \zo^{\overline{X}}$ such that $f(u) \neq f(u_X \circ w)$, it is straightforward, using binary search, to identify an $X_\ell, \ell \not \in I$, and a string $v \in \{0,1\}^n$ with $f(v) \neq f(w_{X_\ell} \circ v_{\overline{X_\ell}})$, using $O(\log r)$ queries to $f$.
The (simple and standard) idea is as follows: Let $R=[r] \setminus I$ be  indices of blocks outside $X$, and let $a=u, b=u_X \circ w$. We repeat the following until $|R|=1$:
\begin{quote}
Let $R'$ be the first half of $R$, and $Y=\cup_{\ell \in R'} X_\ell$.
Query $c=u_{[n] \setminus Y} \circ w_{Y}$ (this is $u$ with blocks in $S'$ set to $w$). If $f(a) \neq f(c)$, we set $R=R'$ and $b=c$. Else we set $R=R \setminus R'$ and $a=c$. 
\end{quote}
(See Subroutine 2 of \cite{LCSSX19} for an detailed implementation of the above sketch.)

\section{Soundness and completeness analysis of the junta subclass tester:  \newline Proof of \Cref{thm:junta-subclass-main} and \Cref{cor:subclass}}
\label{sec:end-of-proof}
 {Throughout \Cref{sec:end-of-proof} , we assume that $0 < \epsilon < 1/2$.} 
\subsection{The analysis when $f$ is a $k$-junta}\label{subsec:yescase}
Throughout this subsection we assume that $f:\{0,1\}^n\rightarrow \{0,1\}$ is a $k$-junta in $\calC(k)$,  and let
 $J$ be its set
  of relevant variables with $|J|\le k$.
The main result of this subsection is \Cref{thm:yes-case}, which establishes that \Cref{alg:junta-subclass-tester} accepts in this case with probability at least $2/3.$
  
\subsubsection{Phases 1 and 2}

We start by noting that most likely, every block $\bX_\ell$ 
  contains at most one variable in $J$:

\begin{lemma}\label[lemma]{lem: one var per block}
In Phase 1, with probability at least $1-1/20$, we have $|\bX_i\cap J|\le 1$ for all $i\in [r]$.
\end{lemma}
\begin{proof}
    Fix any $j_1,j_2\in J$. The probability that they lie in the same block is $1/r$. By a union bound and the assumption that $|J|\le k$, the probability of $|\bX_i\cap J|>1$ for some $i$ is at most $${k \choose 2} \cdot \frac{1}{r} \leq \frac{1}{20}. 
    $$
\end{proof}
From now on we fix the partition   
  $X_1,\ldots,X_r$ of $[n]$ and assume that $|X_i\cap J|\le 1$ for all $i\in [r]$. 
The main lemma for Phase 1 is as follows:

\begin{lemma}\label[lemma]{lemma: testSet-good}
Assume that $f\in \calC(k)$ and  $|X_i\cap J|\le 1$ for all $i$. Then \Cref{alg:junta-subclass-tester} always reaches Phase 3 and when this happens,  
  $f^\ell$ is a literal for every $\ell\in \bI$.
Moreover,
  with probability at least $1-1/20$, \Cref{alg:junta-subclass-tester} reaches Phase 3 with $\bI$ and $\bX$ satisfying the following condition:
\begin{equation}\label{eq:temp202}
\Prx_{\substack{\by \sim f^{-1}(1) \\ \bw \sim \{0,1\}^{\overline{{\bX}}}}}\big[f(\by_{\bX} \circ \bw)=0\big] \le \kappa/4.
\end{equation}
\end{lemma}

\begin{proof}
    For each $\ell$ added to $\bI$, the algorithm has found a $v^{\ell}$ such that $f^\ell$ is not a constant function and thus, $X_\ell$ contains exactly one relevant variable in $J$ and $f^\ell$ must be a literal of that variable.
    
    Given that $|J|\le k$, $\bI$ can never be
    larger than $k$ and thus, \Cref{alg:junta-subclass-tester} always reaches Phase~2. Using \Cref{thm: uniform junta blais} and the fact that 
      $\smash{f^\ell(x)=\overline{f^\ell(\overline{x})}}$ for any string $x$ when $f^\ell$ is a literal,
      we have that \Cref{alg:junta-subclass-tester} always reaches Phase 3.

Now for \Cref{alg:junta-subclass-tester} to violate \Cref{eq:temp202} when
  reaching Phase 3, %
  it must be the case that it has drawn $\bu\sim f^{-1}(1)$ and $\smash{\bw\sim \{0,1\}^{\overline{\bX}}}$ for $T_2=(5\log(20k+1))/\kappa$ many times, and  $f(\bu_{\bX}\circ \bw) =1$ every time.
But given \Cref{eq:temp202}, the probability that this happens is at most 
    $$
    \big(1-\kappa/4\big)^{T_2}\le e^{-(5/4)\log(20k+1)}\le 1/(20k).
    $$
By a union bound (over the at most $k$ 
  different $\bX$'s that are in play across all executions of line {4})  
  the probability of  \Cref{alg:junta-subclass-tester}
  reaching Phase 3 while violating \Cref{eq:temp202} is at most $1/20$. 
\end{proof}

In the rest of the proof, we fix $I$ and $X$ and assume 
  the following:

\begin{assumption}\label[assumption]{assumption:good_case}
\Cref{alg:junta-subclass-tester} reaches Phase 3 
  with $I$ and $X$ that satisfies the following conditions:
\begin{flushleft}\begin{enumerate} 
\item For each $\ell\in I$, we have $|X_\ell\cap J|=1$, and 
  denote the unique element in $X_\ell\cap J$ by $\tau(\ell)$;
\item Every $f^\ell$, $\ell\in I$, is a literal that is  
  either $x_{\tau(\ell)}$ or $\overline{x_{\tau(\ell)}}$.
\item \Cref{eq:temp202} holds.
\end{enumerate}\end{flushleft}
\end{assumption}

Let $h=|I|$ and $\ell_1<\ldots<\ell_h$ be elements 
  in $I$.
Let $\sigma:[h]\rightarrow [n]$ be the 
  injective map defined by $\sigma(i)=\tau(\ell_i)$, and 
  let $S$ be $\{\tau(\ell):\ell\in I\}$ (or the image of $\sigma$) with $|S|=h$.

Let $\sfg:\{0,1\}^h\rightarrow \{0,1\}$ be the following $h$-variable function:
\begin{equation}\label{eq:defg}
\sfg(z):= \text{arg}\max_{b \in \{0,1\}}\left\{ \Prx_{\bw \sim \{0,1\}^{\overline{S}}} \big[f(\sigma(z) \circ \bw))=b\big]\right\}.
\end{equation} 
  To see that $\sfg \in \Apr(h,\kappa)$ whenever $\reldist(f, \sfg_{\sigma}) \leq \kappa$, recall that $J$ is the set of relevant variables of $f$. So let
  $\{\alpha_1, \ldots, \alpha_{{|J|-|I|}}\}$ be the elements of $J \setminus S$. And denote $[n] \setminus J$ as $\{\beta_1, \ldots, \beta_{{n-|J|}}\}$.
  Consider the permutation $\pi:[n] \to [n]$ defined as follow:
  $$
  \pi(j)=\begin{cases}
      \tau(\ell_j) \text{ if } 1 \leq j \leq |I| \\
      \alpha_{j-|I|} \text{ if } |I|+1 \leq j \leq |J|\\
      \beta_{j-|J|} \text{ if } |J|+1 \leq j \leq n
  \end{cases}
  $$
Since $\mathcal{C}(k)$ is closed under permutation $f_\pi \in \mathcal{C}(k)$, more specifically we have $f_{\pi} \in \mathcal{C}(k)^*$ (since $f_\pi$ only depends on the first $|J|$ variables). So letting $\sigma':[h] \to [n]$ be the injective map such that $\sigma'(i)=i$. We have $(\sfg_\sigma)_\pi=\sfg_{\sigma'}$, we have 
$$
\sfg_{\sigma'}(z)= \text{arg}\max_{b \in \{0,1\}}\left\{ \Prx_{\bw \sim \{0,1\}^{n-h}} \big[f_\pi(z_{[h]} \circ \bw)=b\big]\right\}.
$$
By \Cref{obs:permutation_dist} we have $\reldist(f_\pi,\sfg_{\sigma'})=\reldist(f,\sfg_\sigma) \leq \kappa$. Finally since $\sfg(z)=\sfg_{\sigma'}(z \circ 0^{n-h})$, by the construction of $\Apr(h,\kappa)$, we indeed have that $\sfg\in \Apr(h,\kappa)$. 

The lemma below shows that condition 3 in \Cref{assumption:good_case} 
  (\Cref{eq:temp202}) implies $\reldist(f,\sfg_\sigma)\le \kappa$.

\begin{lemma}\label[lemma]{lem:apple}
If $\reldist(f,\sfg_\sigma)> \kappa$, then we have 
    $$\Prx_{\substack{\bu \sim f^{-1}(1) \\ \bw \sim \{0,1\}^{\overline{X}}}}\big[f(\bu_X \circ \bw)=0\big] > \kappa/4.$$
\end{lemma} 
We delay the proof to \Cref{sec:apple}. 
We can now replace condition 3 in \Cref{assumption:good_case}
  by
\begin{flushleft}\begin{enumerate}
\item[$3'$.] The function $\sfg$ defined in \Cref{eq:defg} 
  satisfies $\reldist(f,\sfg_\sigma)\le \kappa$ and is in $\Apr(h,\kappa).$
\end{enumerate}\end{flushleft}

\subsubsection{Phases 3 and 4}

Assume that \Cref{alg:junta-subclass-tester} reaches
  Phase 3 with $I$ and $X$ satisfying \Cref{assumption:good_case} with condition~$3'$.
We first show that
$\sfg\in \Apr(h,\kappa)$ survives the loop on
  lines 14--20 with high probability:

\begin{lemma}\label[lemma]{lem: trimSet not empty}
Under \Cref{assumption:good_case}, with probability at least $1-1/20$, $\sfg$ remains in $\bA$ when \Cref{alg:junta-subclass-tester} reaches line 21. %
\end{lemma}

\begin{proof}
By \Cref{lem: FindBlockValue: bad},  the $\bv$ on line 19 is always $\sigma^{-1}(\bu)$ with $\bu\sim f^{-1}(1)$ and thus $\sfg(\bv)=\sfg_\sigma(\bu)$.
Given that  $\reldist(f,\sfg_\sigma)\le \kappa$ by \Cref{assumption:good_case}, 
  we have that each iteration of the loop removes $\sfg$ from $\bA$ with probability at most $\kappa$.
    Hence, by a union bound over all $T_1=1/(20\kappa)$ loops, the probability that $\sfg$ is removed is at most $T_1\kappa = 1/20$.
\end{proof}

\begin{lemma}\label[lemma]{lem: functions in S are good}
Under \Cref{assumption:good_case}, with probability at least $1-1/20$, on line 21 every 
  $g\in \bA$ 
  has
    $$\Prx_{\bu \sim f^{-1}(1)}\big[g_\sigma(\bu) =0\big] \leq \epsilon/500. $$
\end{lemma}
\begin{proof}
By 
\Cref{lem: FindBlockValue: bad},  the $\bv$ on line 19 is always $\sigma^{-1}(\bu)$ with $\bu\sim f^{-1}(1)$ and thus $g(\bv)=g_\sigma(\bu)$.
Let $g$ be any function in $\Apr(h,\kappa)$ with
$$
\Prx_{\bu \sim f^{-1}(1)}\big[g_\sigma(\bu)=0\big] \geq \epsilon/500.
$$ 
Each iteration of the loop removes $g$ with probability at least $\eps/500$ and thus,
  $g$ survives with probability at most
$ 
(1-\eps/500)^{T_1}.
$ 
Given the choice of 
$$
T_1=O\left(\frac{\log |\calC(k)^* |}{\eps}\right),
$$
by making the hidden constant large enough, $g$ survives with
  probability at most $1/(20|\calC(k)^*|)$. So the lemma 
  follows from \Cref{labelfact} and a union bound over all functions $g$.
\end{proof}
Hence we get the following corollary about the function $\frakg$ passed down to Phase 4:
 
\begin{corollary}\label[corollary]{cor:rd_fA_f_small}
Under  \Cref{assumption:good_case}, with probability at least $1-1/10$, 
  \Cref{alg:junta-subclass-tester} reaches Phase 4 with $\frakg\in \Apr(h,\kappa)$ that satisfies
\begin{equation}\label{eq:temp99}\Prx_{\substack{\bz \sim {\frakg}^{-1}(1)\\ \bw\sim \{0,1\}^{\overline{S}}}}\big[f(\sigma(\bz)\circ \bw)=0\big] \leq \epsilon/400.
\end{equation}
\end{corollary}
\begin{proof}
Assume that the conclusions of both \Cref{lem: trimSet not empty} and 
  \Cref{lem: functions in S are good} hold.
So $\frakg$ is picked and satisfies 
\begin{equation} \label{eq:likeable-equation}
|\frakg_\sigma^{-1}(1)|\le |\sfg_\sigma^{-1}(1)|\quad\text{and}\quad
\frac{|f^{-1}(1)\setminus \frakg_\sigma^{-1}(1)|}{|f^{-1}(1)|}\le \eps/500.
\end{equation}
Let $N=|f^{-1}(1)|$. Since $\reldist(f,\sfg_\sigma)\le \kappa$, we have that
$$
|\frakg_\sigma^{-1}(1)|\le |\sfg_\sigma^{-1}(1)|\le (1+\kappa)N.
$$
On the other hand, from the second part of \Cref{eq:likeable-equation} we have that 
$$
|\frakg_\sigma^{-1}(1)|\ge |f^{-1}(1)\cap \frakg_\sigma^{-1}(1)|\ge (1-\eps/500)N.
$$
As a result, the probability on the left-hand side of \Cref{eq:temp99} is 
$$
\frac{|\frakg_\sigma^{-1}(1)\setminus f^{-1}(1)|}{|\frakg_\sigma^{-1}(1)|}
\le \frac{(1+\kappa)N-(1-\eps/500)N}{(1-\eps/500)N}\le \eps/400,
$$
where we used the choice that $\kappa$ can be made sufficiently small
  compared to $\eps$ by making the hidden constant in $T_1$ large enough. 
\end{proof}

Finally we show that Phase 4 accepts with high probability:
\begin{lemma}\label[lemma]{lem: Verify accepts w.h.p}
Under \Cref{assumption:good_case} and assuming that  
  \Cref{alg:junta-subclass-tester} reaches Phase 4 with a function $\frakg$ that satisfies \Cref{eq:temp99},
then  \Cref{alg:junta-subclass-tester} rejects in Phase 4 with probability at most $1/20$.
\end{lemma}

\begin{proof}
By \Cref{lem:mapback} and the observation made in the overview, 
   $\bu$ on line 24 has the same distribution as 
  $\sigma(\bz)\circ \bw$ where $\smash{\bw\sim \{0,1\}^{\overline{S}}}$.
As a result, \Cref{alg:junta-subclass-tester} rejects with probability
  at most $1/20$ by a union bound over the $20/\eps$ many loops.
\end{proof}

\subsubsection{Putting the pieces together}

We summarize the above step-by-step analysis with the following theorem:

\begin{theorem} \label{thm:yes-case}
    If $f \in \mathcal{C}(k)$, then \Cref{alg:junta-subclass-tester} accepts $f$ with probability at least $2/3$.
\end{theorem}
\begin{proof}
    By \Cref{lem: one var per block} and \Cref{lemma: testSet-good}, with probability at least $1-1/10$ 
    \Cref{alg:junta-subclass-tester} reaches Phase 3 with $I,X$ satisfying 
\Cref{assumption:good_case}.
    When this happens, by \Cref{cor:rd_fA_f_small} and \Cref{lem: Verify accepts w.h.p}, \Cref{alg:junta-subclass-tester} passes through Phases 3 and 4 and 
    accepts with probability at least $1-3/20$.
    Thus, the probability that the algorithm rejects is at most $1/10 + 3/20 < 1/3$.
\end{proof}

\subsection{The analysis when $f$ is far from $\calC(k)$.}\label{subsec:nocase}
Throughout this subsection we assume that $f$ is $\eps$-far from every function in $\calC(k)$ in relative distance,
and we show in \Cref{thm:no-case} that \Cref{alg:junta-subclass-tester} rejects $f$ with probability at least 2/3.

\subsubsection{Phases 1 and 2}

Fix any partition $X_1,\ldots,X_\ell$ of $[n]$.
Clearly $|\bI|\le k$ when \Cref{alg:junta-subclass-tester} reaches Phase 2; otherwise $f$ would have been rejected in Phase 1. 

The lemma below shows that it is unlikely for \Cref{alg:junta-subclass-tester} to reach Phase 3 if
  any $f^\ell$, $\ell\in \bI$, is far from every literal under the uniform distribution:

\begin{lemma}\label[lemma]{lem:all_f^l_are_close}
The probability of \Cref{alg:junta-subclass-tester} reaching
  Phase 3 with $\bI$ such that $f^\ell$ for some $\ell\in \bI$ is 
  $(1/30)$-far from every literal under the uniform distribution is at most $1/15$.
\end{lemma}

\begin{proof}
    If $f^{\ell}$ is $1/30$-far from every literal with respect to the uniform distribution then it is either \begin{enumerate}
        \item $1/30$-far from every $1$-Junta,
        \item or $1/30$ far from every literal and $1/30$ close to a $0$-Junta (i.e.~the constant-0 or constant-1 function).
    \end{enumerate}
     In case $1$, by \Cref{thm: uniform junta blais}, with probability at least $14/15$, $\textsc{UniformJunta}(f^\ell, 1, 1/30, 1/15$) rejects, in which case \Cref{alg:junta-subclass-tester} rejects. In case $2$,  $f^{\ell}$ is $(1/30)$-close to a constant function; say w.l.o.g.~that it is close to the constant-$0$ function. The probability we fail to reject on line $11$ is :

    \begin{align*}
        \Prx_{\bb \sim \{0,1\}^{X_\ell}}\big[f^\ell(\bb) \neq f^\ell(\overline{\bb})\big] %
         \leq \Prx_{\bb \sim \{0,1\}^{X_\ell}}\big[ f^\ell(\overline{\bb}) =1\big]  + \Prx_{\bb \sim \{0,1\}^{X_\ell}}\big[f^\ell(\bb)=1\big]  
         \leq 1/15,
    \end{align*}
    where the last inequality follows from the fact $f^{\ell}$ is $(1/30)$-close to the constant-$0$ function under the uniform distribution and $\bb$, $\overline{\bb}$ are both sampled uniformly at random.
\end{proof}

Fix $I$ and $X$ with $|I|=h\le k$, and
let $\ell_1 <\cdots< \ell_{h}$ be the elements of $I$.
We make the following assumption in the rest of this subsection:

\begin{assumption}
\label[assumption]{assumption:bad_case}
\Cref{alg:junta-subclass-tester} reaches Phase 3 with $I$ and $X$ satisfying the following condition: 
\begin{enumerate}\item[]For every $\ell\in I$, $f^{\ell}$ is $(1/30)$-close to a literal $x_{\tau(\ell)}$ or $\overline{x_{\tau(\ell)}}$ under the uniform distribution. \end{enumerate}
\end{assumption}

Similar to the previous subsection, we let 
  $\sigma:[h]\rightarrow [n]$ be the injective map with
  $\sigma(i)=\tau(\ell_i)$ for each $i\in [h]$, and 
  let $S=\{\tau(\ell):\ell\in I\}$ with $|S|=h$. 

\subsubsection{Phases 3 and 4}

Under \Cref{assumption:bad_case}, we show in the next lemma that most likely, every function $g\in \Apr(h,\kappa)$ that survives in $\bA$ {up to line~21} satisfies that $g(\sigma^{-1}(\bu))=0$ with low probability when $\bu\sim f^{-1}(1)$:

\begin{lemma}\label[lemma]{lem:trim_bad_case}
Under \Cref{assumption:bad_case}, with probability at least $1-1/20$, when \Cref{alg:junta-subclass-tester} reaches {line 21}
 every $g\in \bA $ satisfies%
    $$\Prx_{\bu \sim f^{-1}(1)}\big[g_\sigma(\bu)=0\big]\le  \epsilon/500.$$
\end{lemma}
\begin{proof}
The proof is very similar to that of \Cref{lem: functions in S are good}.
The only difference is that it is not always the case that $\bv=\sigma^{-1}(\bu)$;
  by a union bound over the at most $k$ calls to $\RelV$ (as well~as the choice of $\delta=1/(20k)$), we have $\bv=\sigma^{-1}(\bu)$
  with probability at least $19/20$. Therefore,
  any $g$ that violates the inequaility above would be removed 
  with probability at least $(19/20) \eps/500$ in~each loop.
The rest of the proof is the same as that of \Cref{lem: functions in S are good}.
\end{proof}

Assuming the conclusion of \Cref{lem:trim_bad_case} holds, we have that either \Cref{alg:junta-subclass-tester} rejects because $\bA$ is empty at the end of Phase 3 (in which case we are done), or the function $\frakg \in \Apr(h,\kappa)$ passed down to Phase 4 satisfies  
$$\Prx_{\bu \sim f^{-1}(1)}\big[\frakg_\sigma (\bu) =0\big]\le   \epsilon/500.$$
Letting $N=|f^{-1}(1)|$, we have that
$$
 |f^{-1}(1)\setminus \frakg_\sigma^{-1}(1)|\le  \eps N/500.
$$
On the other hand, given that $\frakg\in \Apr(h,\kappa)$, by our assumption in this subsection that $\reldist(f,\calC(k) >\eps$ and \Cref{lem:simple2}
  (and the fact that $\kappa$ is sufficiently smaller than $\eps$), we have 
$
\reldist(f,\frakg_\sigma)\ge \eps/2
$,
and we hope to show that
\begin{equation}\label{eq:temp303}
\frac{|\frakg_\sigma^{-1}(1)\setminus f^{-1}(1)|}{|\frakg_\sigma^{-1}(1)|}\ge \eps/7.
\end{equation}
To this end, consider two cases.
(1) If $|\frakg_\sigma^{-1}(1)|\ge 2\cdot |f^{-1}(1)|$, then the ratio above is at least 
$$
\frac{|\frakg_\sigma^{-1}(1)|-|f^{-1}(1)|}{|\frakg_\sigma^{-1}(1)|}\ge \frac{|\frakg_\sigma^{-1}(1)|/2}{|\frakg_\sigma^{-1}(1)|}\ge 1/2 > \eps/7.
$$
(2) If $|\frakg_\sigma^{-1}(1)|< 2\cdot |f^{-1}(1)|$, then the ratio is at least
$$
\frac{|\frakg_\sigma^{-1}(1) \ \triangle \ f^{-1}(1)|-|f^{-1}(1)\setminus \frakg_\sigma^{-1}(1)|}{2N}\ge \frac{\eps N/2-\eps N/500}{2N}\ge \eps/7.
$$
Finally, we show that \Cref{alg:junta-subclass-tester} rejects in 
  Phase 4 with high probability:
\begin{lemma}\label[lemma]{lem:verify_bad_case}
   Under \Cref{assumption:bad_case}, assume that \Cref{alg:junta-subclass-tester} reaches Phase 4 with $\frakg$ satisfying \Cref{eq:temp303}. Then Phase 4 rejects with probability at least $0.8 $.
\end{lemma}
\begin{proof}
Recall the observation in the overview that $\by_{\sigma,\bz}$
  has the same distribution as $\sigma(\bz)\circ \bw$ with
  $\smash{\bw\sim \{0,1\}^{\overline{S}}}$.
Hence the probability of $f(\by_{\sigma,\bz})=0$ for at least 
  one loop, by  \Cref{eq:temp303}, is at least
$$
1-(1-\eps/7)^{20/\eps}\ge 1-e^{-20/7}.
$$
For that loop, by \Cref{lem:mapback}, we have 
  $\bu=\by_{\sigma,\bz}$ with probability at least $19/20$,
  in which case $f(\bu)$ $=0$ and \Cref{alg:junta-subclass-tester} rejects.
Overall \Cref{alg:junta-subclass-tester} rejects with probability at least $$(19/20)(1-e^{-20/7})\ge 0.8 . $$ %

\end{proof}

\subsubsection{Putting the pieces together}
\begin{theorem} \label{thm:no-case}
    If $f:\zo^n \to \zo$ is $\epsilon$-far from every function in $\mathcal{C}(k)$ in relative distance, then \Cref{alg:junta-subclass-tester} rejects with probability at least $2/3$.
\end{theorem}
\begin{proof}
By a union bound over the bad events in 
  \Cref{lem:all_f^l_are_close}, \Cref{lem:trim_bad_case}
  and \Cref{lem:verify_bad_case}, we have that 
  \Cref{alg:junta-subclass-tester} accepts with probability at most
  $1/15+1/20+0.2<1/3.$
\end{proof}

\subsection{Query complexity of \Cref{alg:junta-subclass-tester}}\label{subsec:querycomp}

Finally, we show that the query complexity of \Cref{alg:junta-subclass-tester}
  (the number of random samples drawn from $f^{-1}(1)$ and 
    membership queries on $f$) is  at most $\smash{\tilde{O}\pbra{{  {(k/\eps) \log |\calC(k)^*|} }}}$.

Recall the choice of $T_2$ in \Cref{alg:junta-subclass-tester}.
 In Phase~1, we repeat lines 3--{7} until either $|\bI| >k$~or $t=T_2$.
In one iteration of the loop, if the ``if'' condition on line {4} is true, then lines {5--7}~are executed which makes $O(\log r)=O(\log k)$ queries, and $|\bI|$ increases by one and $t$ is set to $0$.  
And if the ``if'' condition on line {4} is false, then line {4} only makes two queries and $|\bI|$ does not increase, but $t$ increases by one.
Therefore, the query complexity of this phase is 
$$O\pbra{k(T_2 + \log k)}= {\tilde{O}}\pbra{(k/\eps) \log|\calC(k)^*|}.$$ 
For Phase 2, given that $|\bI| \leq k$, the query complexity is $O(k)$ by %
\Cref{thm: uniform junta blais}. %

In Phase~3, given that the number of queries made by each call to $\RelV$ is $O(\log k)$ and $h\le k$,
  the query complexity is $$O\pbra{kT_1\log k}={\tilde{O}}\pbra{(k/\eps) \log|\calC(k)^*|}.$$ 
Given that $\MapBack$ makes $O(k\log k)$ queries, the
  query complexity of Phase 4 is $O((k/\eps)\log k)$.
 Putting things together, the query complexity of 
 \Cref{alg:junta-subclass-tester} is  $\smash{\tilde{O}\pbra{{ {(k/\eps) \log |\calC(k)^*|} }}}$ as claimed.

\subsection{Proof of \Cref{lem:apple}}\label{sec:apple}

\begin{proof}
Recall that in the context of \Cref{lem:apple}, $f$ is a $J$-junta, $S=X\cap J$, and each block of $X$ contains exactly one relevant variable. 
We consider the following equivalent way of drawing $\bu_X\circ \bw$ with $\bu\sim f^{-1}(1)$ and $\smash{\bw\sim \{0,1\}^{\overline{X}}}$: draw $\bu\sim f^{-1}(1)$ and $\smash{\bw\sim \{0,1\}^{\overline{S}}}$ and return $\bu_S\circ \bw$.
The two distributions are the same because every variable 
  in $X\setminus S$ is irrelevant. As a result, 
 \begin{equation}\label{eq:temp404}
  \Prx_{\substack{\bu \sim f^{-1}(1) \\ \bw \sim \{0,1\}^{\overline{X}}}}
  \big[f(\bu_{X} \circ \bw)=0\big] =      \Prx_{\substack{\bu \sim f^{-1}(1) \\ \bw \sim \{0,1\}^{\overline{S}}}}\big[f(\bu_{S} \circ \bw)=0\big]. 
    \end{equation} 

Let $N=|f^{-1}(1)|$, $F=f^{-1}(1)$ and $G=\sfg_\sigma^{-1}(1)$.
Given that $|F \ \triangle \  G|\ge \kappa N$, we
  consider the two cases of $|F\setminus G|>\kappa N/2$
    and $|G\setminus F|>\kappa N/2.$
    
For the case when $|F\setminus G|>\kappa N/2$, we have
$$\Prx_{\substack{\bu \sim f^{-1}(1)\\ \bw \sim \zo^{\overline{S}} }}\big[f(\bu_{S} \circ \bw)=0\big] \geq 
\Prx_{\bu \sim f^{-1}(1) } \big[\bu \in F\setminus G\big] \cdot \Prx_{\substack{\bu \sim f^{-1}(1)\\ \bw \sim \zo^{\overline{S}}} }\big[f(\bu_{S} \circ \bw)=0 \hspace{0.08cm}|\hspace{0.08cm} \bu \in F\setminus G \big].$$
The first probability on the RHS is at least $\kappa/2$;
  the second probability on the RHS is at least $1/2$ given the definition of
  $\sfg$ from $f$ and the condition that $\bu\in F\setminus G$ and thus, $\sfg_\sigma(\bu)=0$.

For the second case when $|G \setminus F| \geq \kappa N/2$, consider any fixed $a \in G \setminus F$. Since $a \in G$, we have $$\Prx_{\bw \sim \zo^{\overline{S}}}\big[f(a_{S} \circ \bw)=1\big] \geq 1/2$$ and thus there must exist $2^{|\overline{S}|}/2$ many $y \in F$ with $y_{S}=a_S$. So 
$$\Prx_{\substack{\bu \sim f^{-1}(1)\\ \bw \sim \{0,1\}^{\overline{S}}}}\big[\bu_{S} \circ \bw = a\big] \geq \Prx_{\bu \sim f^{-1}(1)}\big[\bu_{S} = a_{S}\big] \cdot \Prx_{\bw \sim \zo^{\overline{S}}}\big[\bw=a_{\overline{S}}\big] \geq \frac{{}2^{|\overline{S}|}}{2|N|}\times 2^{-|\overline{S}|} \geq \frac{1}{2|N|}.$$
As a result, we have $$\Prx_{\substack{\bu \sim f^{-1}(1)\\ \bw \sim \{0,1\}^{\overline{S}}}}\big[f(\bu_{S} \circ \bw)=0\big] \geq \sum_{a \in G \setminus F} \left(\Prx_{\substack{\bu \sim f^{-1}(1)\\ \bw \sim \{0,1\}^{\overline{S}}}}\big[\by_{S} \circ \bw=a\big]\right) \geq \frac{\kappa|N|}{2} \cdot \frac{1}{2|N|} = \frac{\kappa}{4}.$$
This finishes the proof of the lemma.
\end{proof}

\subsection{Applications of \Cref{thm:junta-subclass-main}:  Proof of \Cref{cor:subclass}}\label{sec:subclassapplication}

For the class $\calC(k)$ of size-$k$ decision trees, we observe that every size-$k$ decision tree is a $k$-junta and we use the fact that $|\calC(k)^*|$, the number of size-$k$ decision trees over $x_1,\dots,x_k$, is at most $k^{O(k)}$ (see \cite{DLM+:07} for the simple proof).  The claimed testing result follows immediately by applying \Cref{thm:junta-subclass-main}. 
A similar argument is used for the class of size-$k$ branching programs, using the fact that there are at most $k^{O(k)}$ size-$k$ branching programs, and likewise for the class of size-$k$ Boolean formulas; again see \cite{DLM+:07} for the simple counting arguments giving these upper bounds. (Here we are using the standard definition that the size of a Boolean formula is the number of leaves.) %

Finally, we note that \Cref{thm:junta-subclass-main} can be applied to obtain efficient relative-error testing algorithms for a number of other subclasses of juntas that were studied in the uniform-distribution setting; as detailed in \cite{Bshouty20}, these include $s$-term monotone $r$-DNF, $s$-term unate $r$-DNF, length-$k$-decision list, parities of length at most $k$, conjunctions of length at most $k$, $s$-sparse polynomials over $\F_2$ of degree $d$, and functions with Fourier degree at most $d$. We leave the straightforward application of \Cref{thm:junta-subclass-main} to these classes as an exercise for the interested reader.

\section{Testing Juntas with Relative Error:  Proof of \Cref{thm:junta-main}} \label{sec:juntas}

In this section we present a two-sided $\tilde{O}(k/\eps)$-tester for $k$-juntas under relative distance. {We assume that $0<\epsilon<1/2$}. 

Before presenting the tester, we prove a lemma (\Cref{lem:keyjunta})
  essential for both the design and analysis of our tester. 
We start with a relative-distance analogue of Lemma~3.2 of \cite{LCSSX19}; in words, it says that if $\reldist(f,g)$ is large for every $k$-junta $g$, then for any fixed set $J$ of at most $k$ variables, rerandomizing all variables outside of $J$ in a random satisfying assignment of $f$ is fairly likely to make it into an unsatisfying assignment of $f$.

\begin{lemma} \label[lemma]{lem:3.2-analogue}
    Let $f:\{0,1\}^n\rightarrow \{0,1\}$ be a Boolean function such that $\reldist(f,g) \geq \eps$ for every $k$-junta $g$.  Then for any $J \subseteq [n]$ with $|J|\leq k$, we have 
    $$
    \Prx_{\substack{\bu \sim f^{-1}(1)\\ \bw \sim \{0,1\}^{\overline{J}}}}\big[f(\bu_J \circ \bw)=0\big] =
    \Prx_{\substack{\bu \sim f^{-1}(1)\\ \bw \sim \{0,1\}^{\overline{J}}}}\big[f(\bu) \neq f(\bu_J \circ \bw)\big] \geq \epsilon/4.$$
 \end{lemma}
\begin{proof} 
    Consider the function $h: \{0,1\}^n \to \{0,1\}$ defined as follows:
    $$h(x)= \text{arg}\max_{b \in \{0,1\}}\left\{ \Prx_{\bw \sim \{0,1\}^{\overline{J}}} \big[f(x_J \circ \bw))=b\big]\right\}$$
    where ties are broken arbitrarily.
Since $h$ is a $J$-junta, letting $N:= |f^{-1}(1)|$, we have
$$
\big|f^{-1}(1)\ \triangle\ h^{-1}(1)\big|\ge \eps N.
$$
Let $Z_1$ be the set of $z\in \{0,1\}^n$ with  
  $f(z)=1$ and $h(z)=0$, and let 
$Z_2$ be the set of $z\in \{0,1\}^n$ with $f(z)=0$ and $h(z)=1$.
Then either $Z_1$ or $Z_2$ is of size at least  $\eps N/2$.

In the first case ($|Z_1|\ge \eps N/2$), we have
    \begin{align*}
        \Prx_{\substack{\bu \sim f^{-1}(1) \\ \bw \sim \{0,1\}^{\overline J}}}\big[ f(\bu_J \circ \bw)=0\big]\ge  \frac{1}{N} \sum_{z \in Z_1}  \Prx_{\bw \sim \{0,1\}^{\overline J}}\big[  f(z_J \circ \bw)=0\big] 
        \geq \frac{1}{N}\cdot |Z_1|\cdot \frac{1}{2}  \geq \epsilon/4,
    \end{align*}
where the second inequality used $h(z)=0$ and the definition of $h$.

In the second case ($|Z_2|\ge \eps N/2$), 
given the definition of $Z_2$ we have
\begin{equation}\label{eq:temp1}
\Prx_{\substack{\bu \sim f^{-1}(1)\\ \bw \sim \{0,1\}^{\overline{J}}}}\big[f(\bu_J \circ \bw)=0\big]
\ge \sum_{z\in Z_2} \Prx_{\substack{\bu \sim f^{-1}(1)\\ \bw \sim \{0,1\}^{\overline{J}}}}\big[ \bu_J \circ \bw =z\big].
\end{equation}
Fixing any
  point $z\in Z_2$, we have
\begin{equation}\label{eq:temp2}
\Prx_{\substack{\bu \sim f^{-1}(1)\\ \bw \sim \{0,1\}^{\overline{J}}}}\big[ \bu_J \circ \bw=z\big]
\ge \frac{2^{|\overline{J}|-1}}{N}\cdot \frac{1}{2^{|\overline{J}|}}\ge \frac{1}{2N},
\end{equation}
where $2^{|\overline{J}|-1}$ lowerbounds the number of $x\in f^{-1}(1)$ with $x_J=z_J$ given the definition 
  of $h$.
The second case then follows by combining $|Z_2|\ge \eps N/2$, \Cref{eq:temp1}, and \Cref{eq:temp2}.
\end{proof}

We are now ready to prove \Cref{lem:keyjunta}.
\begin{lemma} \label[lemma]{lem:keyjunta}
Let $f:\{0,1\}^n\rightarrow \{0,1\}$ be a Boolean function that is $\epsilon$-far from every $k$-junta in relative distance. Let $X \subseteq [n]$ be such that $$\Prx_{\substack{\bu \sim f^{-1}(1) \\ \bw \sim \{0,1\}^{\overline X}}}\big[f(\bu) \neq f(\bu_X \circ \bw)\big] \leq \epsilon/20$$
and let $J$ be a subset of at most $k$ variables from $X$.
Then we have  
$$\Prx_{\substack{\bu \sim f^{-1}(1) \\ \bw \sim \{0,1\}^{\overline X},\\ \by \sim \{0,1\}^{X \setminus J}}}\big[f(\bu_X \circ \bw) \neq f(\bu_J \circ \by \circ \bw)\big] \geq \epsilon/5.$$
\end{lemma}
\begin{proof}
Let $\bu\sim f^{-1}(1)$, $\bw\sim \{0,1\}^{\overline{X}}$ and $\by\sim \{0,1\}^{X\setminus J}$ as above. Then the following inequality completes the proof (where the last line uses the assumption of the Lemma, and \Cref{lem:3.2-analogue}.
\begin{align*}
 \Prx \big[f(\bu_X \circ \bw) \neq f(\bu_J \circ \by \circ \bw)\big] 
    & \geq \Prx \big[f(\bu)=f(\bu_X \circ \bw) \land f(\bu) \neq f(\bu_J \circ \by \circ \bw)\big] \\[0.4ex]
    & \geq \Prx \big[f(\bu)=f(\bu_X \circ \bw) \big]+  \Prx \big[f(\bu) \neq f(\bu_J \circ \by \circ \bw)\big] - 1 \\[0.4ex]
    & \geq 1-\epsilon/20 + \epsilon/4 -1 \\[0.4ex]
    & = \epsilon/5,
    \end{align*}
\end{proof}

\subsection{The tester and an overview}

The junta testing algorithm, called {\textsc{JuntaTester}}, is presented in  \Cref{algo:JuntaTester}.
We bound its sample and query complexity in  \Cref{subsec:query}.
We show in \Cref{subsec:isjunta} that \Cref{algo:JuntaTester} accepts with
  probability at least $2/3$ when $f$ is a $k$-junta, and show in  \Cref{subsec:farfromjunta} that it rejects with probability at least $2/3$ when $f$ is $\eps$-far from every $k$-junta in relative distance.

Before the formal proof, in this subsection we give an overview of the algorithm
  and its analysis.
  Since we are concerned with the entire class of $k$-juntas rather than a subclass, there is no need to do the kind of ``implicit learning'' that was used in \Cref{alg:junta-subclass-tester}; this lets us achieve a much better query complexity than would follow from a naive application of \Cref{alg:junta-subclass-tester}.
Looking ahead, compared to \Cref{alg:junta-subclass-tester},
  the savings comes from the fact that no trimming is needed in \Cref{algo:JuntaTester}.
  Instead, 
  at a high level 
  the oranization of our \Cref{algo:JuntaTester} 
is set up so as to follow the distribution-free junta testing algorithm and proof of correctness from \cite{Bshouty19} quite closely,
 but there are some differences which we describe at the end of this subsection.

Let $f:\{0,1\}^n\rightarrow \{0,1\}$ be the input function that is being tested.
In Phase A,  \Cref{algo:JuntaTester} starts by drawing a partition of
  variables $[n]$ into $r=O(k^2)$ many blocks $\bX_1,\ldots,\bX_r$
  uniformly at random.
In Phase~B,  \Cref{algo:JuntaTester} tries to find as many \emph{relevant} blocks $\bX_\ell$ as possible under a query complexity budget. Here intuitively we say $\bX_\ell$ is a relevant block of $f$
  if we have found a string $v^\ell\in \{0,1\}^n$ such that 
  $f^\ell:\{0,1\}^{\bX_\ell}\rightarrow \{0,1\}$, defined as 
  $$ {f^\ell(x) =f\left(x\circ v^\ell_{\overline{\bX_\ell}}\right)},$$ is not a constant function.
Let $\{\smash{\bX_\ell}\}_{\smash{\ell\in I}}$ be the set of relevant blocks found so far (with $\smash{I=\emptyset}$ at the beginning), and let $\bX$ be their union. 
Phase B draws 
  a uniform satisfying assignment $\bu \sim f^{-1}(1)$ and a random string $\bw$ 
  over $\smash{\{0,1\}^{\overline{\bX}}}$; if $f(\bu_{\bX}\circ\bw)=0$,  then a binary search over blocks can be run
  on $\bu$ and $\bu_{\bX}\circ \bw$, using $O(\log r)=O(\log k)$ many queries, to find a new relevant block $\bX^\ell$ together with an accompanying $v^\ell$.
If Phase B found too many relevant blocks (i.e., $|I|>k$) then \Cref{algo:JuntaTester} rejects; on the other hand, if it fails to make progress by finding a new relevant block after $T=O(\log(k/\eps))$ rounds, \Cref{algo:JuntaTester} moves on to Phase~C. The latter helps ensure that the number of queries used in Phase B stays within the budget of $\tilde{O}(k/\eps)$.

In Phase~C \Cref{algo:JuntaTester} checks whether every function $f^\ell$, $\ell\in I$, is close to a literal under the uniform distribution  (i.e., $x_{\tau(\ell)}$ or $\overline{x_{\tau(\ell)}}$ for some variable $\tau(\ell)\in \bX_\ell$),
and rejects if any of them is not.

When \Cref{algo:JuntaTester} reaches Phase~D, one may assume that the relevant blocks found satisfy the following two conditions:
\begin{flushleft}\begin{enumerate}
\item First, we have\begin{equation}\label{temp22}
\Prx_{\substack{\bu \sim f^{-1}(1)\\ \bw \sim \{0,1\}^{\overline \bX}}}\big[f(\bu_{\bX} \circ \bw) \neq f(\bu)\big] \le \epsilon/20;
\end{equation}
since otherwise it is unlikely for Phase B to fail to find a new 
  relevant block after $T$ repetitions before moving to Phase~C.
\item Second, every $f^\ell$ is close to a literal $x_{\tau(\ell)}$ or $\overline{x_{\tau(\ell)}}$ for some unique $\tau(\ell)\in \bX_\ell$;
since otherwise Phase~C rejects with high probability.
\end{enumerate}
\end{flushleft}
Let $J$ be the set of indices $\{\tau(\ell)\}_{\ell\in I}$, with $|J|\le k$. Then by \Cref{lem:keyjunta}, we have 
\begin{equation}\label{temp33}
\Prx_{\substack{\bu \sim f^{-1}(1) \\ \bw \sim \{0,1\}^{\overline X},\\ \by \sim \{0,1\}^{X \setminus J}}}\big[f(\bu_X \circ \bw) \neq f(\bu_J \circ \by \circ \bw)\big] \geq \epsilon/5
\end{equation}
when $f$ is $\eps$-far from every $k$-junta in relative distance.
While $\bu$ and $\bw$ above are easy to sample, 
  Phase~D uses a delicate subroutine (lines {16--25})
  to sample $\by$.
The challenge here is that \Cref{algo:JuntaTester} 
  does not know $J$ (even though it is well defined given that 
  every $f^\ell$ is close to a literal) and it would require too many queries
  to explicitly identify any $\tau(\ell)$ given that each block is likely to contain $\Theta(n/k^2)$ many variables (recall that we wants to avoid $O(\log n)$-type query complexities).
We explain how the subroutine works in  \Cref{subsec:farfromjunta}; 
\Cref{algo:JuntaTester} rejects in Phase~D if 
\begin{equation} \label{eq:a}
f(\bu_X \circ \bw) \neq f(\bu_J \circ \by \circ \bw)\end{equation}
  and it accepts if this does not happen after $M=O(1/\eps)$ rounds.

For the case when $f$ is a $k$-junta, as will become clear in 
  the analysis in  \Cref{subsec:isjunta}, \Cref{algo:JuntaTester} always accepts if the partition is such 
  that every block $\bX_\ell$ contains at most one relevant variable of $f$. Given that we draw $r=O(k^2)$ blocks, this happens with high probability by a standard birthday paradox analysis.
For the case when every $k$-junta is $\eps$-far from $f$ in relative distance, on the other hand, it follows from \Cref{temp33} (which follows from \Cref{lem:keyjunta} and the analysis of Phase B and C) that  \Cref{algo:JuntaTester} rejects with high probability.

Finally, we briefly discuss how our algorithm and (mostly) its analysis differ from  \cite{Bshouty19}.  Each of our algorithm's phases is quite similar to the corresponding phase of the  \cite{Bshouty19} algorithm. However, in Phase~B,
\cite{Bshouty19} picks a restriction $f'$ of $f$ in which every variable outside the ``relevant blocks'' is set to $0$. Since \cite{Bshouty19} works in the distribution free setting it's easy to check whether $f$ and $f'$ are close, by simply sampling $\bu \sim \mathcal{D}$ and checking if $f(\bu)=f'(\bu)$. Finally, in Phase~D \cite{Bshouty19} checks if $f'$ is close to a specific $k$-junta under the distribution $\mathcal{D}$. In the no case of the distribution-free setting, since $f$ is far from every $k$-junta and $f'$ is close to $f$, $f'$ must fail this final check with high probability.  

In contrast, in our relative-error setting it is not easy to check if $f$ is relative-error close to a restriction like $f'$. In the relative-error setting we can only get samples from $f^{-1}(1)$, and it is not clear how to use such samples to detect that $f$ is far from $f'$. This is because it could be the case that $\reldist(f,f')$ is large, but for $\bu \sim f^{-1}(1)$ we always have $f'(\bu)=1$. 

To deal with this issue, we do not use a restriction of $f$, but rather we set variables in the ``irrelevant blocks'' uniformly at random in both Phase~B and Phase~D.  We do not get that $f$ is close to some function $f'$ by the end of Phase~B; instead, the correctness of our algorithm crucially relies on \Cref{lem:keyjunta}, which does not have an analogue in the analysis of \cite{Bshouty19}.

\subsection{Sample and query complexity of \Cref{algo:JuntaTester}}\label{subsec:query}

We first note that Phase A does not make any samples or queries.
In Phase B, we repeat lines {4--8} until either $|I| >k$ or $t=T$.
In one iteration of the loop, if the ``if'' condition on line {5} is true, then lines {6--7} are executed which makes $O(\log r)=O(\log k)$ queries, and $|I|$ increases by one and $t$ is set to $0$.  
And if the ``if'' condition on line {5} is false, then line {5} only makes two queries and $|I|$ does not increase, but $t$ increases by one.
Therefore, the query complexity of this phase is $O(k(T + \log k))= O(k \log(k/\epsilon))$ given that $T=O(\log(k/\eps))$. Moreover, the number of samples from $f^{-1}(1)$ that are drawn across Phase~B is at most $O(kT)=O(k \log(k/\eps))$, since line~{4} is executed at most $T$ times between every two consecutive increments of $|I|$, and $|I|$ never grows larger than $k$.
In Phase~C, lines {11--12} are executed at most $k$ times, and by \Cref{thm: uniform junta blais}, each call to \textsc{UniformJunta} makes $O(1)$ queries (and no samples), so the query complexity of this phase is $O(k)$.
In the final phase, lines {15--27} are repeated $M=O(1/\epsilon)$ times, and each time the query complexity is $O(kh)=O(k\log (k/\eps))$, so the query complexity of this phase is $O((k/\epsilon) \log(k/\epsilon))$ (and this phase makes $M$ samples, corresponding to the $M$ executions of line~{26}).  Therefore, the total number of queries and samples made by \Cref{algo:JuntaTester} is $O((k/\epsilon)\log(k/\epsilon))$.

\begin{algorithm}
\caption{The relative error junta testing algorithm, \textsc{JuntaTester}$(f,k, \epsilon)$}
\label{algo:JuntaTester}

\Algphase{Phase A: Randomly partition $[n]$ into blocks}

\nonl Parameters used in the algorithm: $r=2k^2$, $T=O(\log (k/\epsilon))$, $M=O(1/\epsilon)$, and $h=O(\log(k/\epsilon))$\; 
 Choose uniformly at random an $r$-way partition $\bX_1, \ldots, \bX_r$ of $[n]$\; %

\Algphase{Phase B: Find relevant blocks}

 Set $\bX=\emptyset$, $I=\emptyset$ and $t=0$\;
\While{$t \neq T$}{

 \hskip0em Draw uniform random $\smash{\bu \sim f^{-1}(1)}$ and $\smash{\bw \sim \{0,1\}^{\overline{\bX}}}$, and set $t \leftarrow t +1$\;
 \If{$\smash{f(\bu_{\bX} \circ \bw) \neq f(\bu)}$}{
 {\hskip0em Binary Search over the elements of $\bX_1,\dots,\bX_r$ that are subsets of $\overline{\bX}$  \hskip0em to find a new relevant set $\bX_\ell$ and a string $v^{\ell} \in \{0,1\}^n$ such that {$$f(v^{\ell}) \neq f\Big(\bw_{\bX_\ell} \circ v^{\ell}_{\overline{\bX_\ell}}\Big).$$}
 \nonl 
 \hskip0em Denote by $f^{\ell} : \{0,1\}^{\bX_\ell} \rightarrow \{0,1\}$ the following function: $$f^{\ell}(x):=f\Big(x \circ v^{\ell}_{\overline{\bX}_\ell}\Big).  $$}

Set $\bX \leftarrow \bX \cup \bX_\ell$, $I \leftarrow I \cup \{\ell\}$ and $t\leftarrow 0$. Reject if $|I| >k$\;

}
}

\Algphase{Phase~C: Test if each $f^\ell$, $\ell\in I$,  is close to a literal}

\ForEach{$\ell \in I$}{
 \hskip0em Run \textsc{UniformJunta}$(f^{\ell} , 1, 1/30, 1/15)$;  Reject if it rejects\;
 \hskip0em Draw $\smash{\bb \sim \{0,1\}^{\bX_\ell}}$; Reject and halt if $f^{\ell}(\bb)= f^{\ell}(\overline{\bb})$\;
}
\Algphase{Phase D: The final test}

\RepTimes{ $M$ }{
 \hskip0em Draw $\smash{\bs \sim \{0,1\}^{\bX}}$ and set $\bz$ to be the all-zero string in $\{0,1\}^{\bX}$\; %
 \ForEach{$\ell \in I$}{
 \hskip0em Set $\smash{\bY_{\ell, \zeta}=\{j \in \bX_\ell \;|\; \bs_j = \zeta\}}$ and $\bG_{\ell,\zeta}=0$ for both $\zeta\in \{0,1\}$\;
 \RepTimes{$h$}{
 \hskip0em Draw  $\smash{ {\bb^\zeta} \sim \{0,1\}^{\bY_{\ell,\zeta}}}$ for both $\zeta\in \{0,1\}$\;
\hskip0em If $\smash{f^{\ell}( {\bb^0} \circ {\bb^1} ) \neq f^{\ell}(\overline{ {\bb^0}} \circ  {\bb^1})}$ then $\bG_{\ell,0} \leftarrow \bG_{\ell,0}+1$\;
\hskip0em If $\smash{f^{\ell}( {\bb^0} \circ {\bb^1}) \neq f^{\ell}({\bb^{0}} \circ \overline{ {\bb^1}} )}$ then $\bG_{\ell,1} \leftarrow \bG_{\ell,1}+1$\;}
 \hskip0em If $\{\bG_{\ell,0}, \bG_{\ell,1}\} \neq \{0,h\}$ then reject\;
 \hskip0em If $\bG_{\ell,0}=h$ then $\bz_{\bX_\ell} \leftarrow  \bs_{\bX_\ell}$ else $\bz_{\bX_\ell} \leftarrow  \overline{\bs_{\bX_\ell}}$\;
 }
 \hskip0em Draw $\smash{\bu \sim f^{-1}(1)}$, $\smash{\bw \sim \{0,1\}^{\overline{\bX}}}$\; If $f(\bu_{\bX} \circ \bw) \neq f((\bu_{\bX} \oplus \bz_{\bX}) \circ \bw)$ then reject\;
  }
Accept\;
\end{algorithm}

\subsection{The analysis when $f$ is a $k$-junta}\label{subsec:isjunta}

Throughout this subsection we assume that $f:\{0,1\}^n\rightarrow \{0,1\}$ is a $k$-junta,  and let
 $J$ be its set
  of relevant variables with $|J|\le k$.
The following lemma implies that after Phase A of  \Cref{algo:JuntaTester}, with probability at least $2/3$, the partition satisfies $|\bX_i \cap J| \leq 1$ for all $i \in [r]$.

\begin{lemma}%
\label[lemma]{lem: Nader 6}
With probability at least $2/3$, we have that $|\bX_i\cap J|\le 1$ for all $i\in [r]$.
\end{lemma}
\begin{proof}
    Fix any $j_1,j_2\in J$. The probability that they lie in the same block is $1/r$. By a union bound and the assumption that $|J|\le k$, the probability of $|\bX_i\cap J|>1$ for some $i$ is at most $${k \choose 2} \cdot \frac{1}{r} \leq \frac{1}{3}. %
    $$
\end{proof}

Next we consider Phases B and C of the algorithm assuming that$|\bX_i\cap J|\le 1$ for all $i\in [r]$.

\begin{lemma}%
\label[lemma]{lem: Nader7}
Assume that 
  $| \bX_i\cap J|\le 1$ for all $i\in [r]$. Then 
 \Cref{algo:JuntaTester} always reaches Phase~D and when this happens, we have  $|\bX_i\cap J|=1$ for all $i\in I$. 
\end{lemma}
\begin{proof}
    On line {7}, for $\ell$ to be added to $I$,  we found $v^{\ell}$ such that $$f(v^{\ell}) \neq f\Big(\bw_{\bX_\ell} \circ v^{\ell}_{\overline{\bX_\ell}}\Big)$$ and thus, $|\bX_\ell\cap J|$ must be exactly $1$. As such, for every $\ell \in I$, we have that the function $f^{\ell}$ is a literal.

    If \Cref{algo:JuntaTester} fails to reach Phase D, then it must have halted on line 7, 11 or 12. If it halts on line 7, then $|I|>k$, but this implies $|J|>k$ which is a contradiction. If it halts on line ${11}$, then since {\textsc{UniformJunta}} has one-sided error it must be the case that $f^{\ell}$ is not a $1$-junta, but we know $f^{\ell}$ is a literal (so it is a $1$-junta). And if halts on line 12, then $f^{\ell}(\bb) = f^{\ell}(\overline{\bb})$ which contradicts the fact that  $f^{\ell}$ is a literal. So in all cases, we reach a contradiction.
\end{proof}

\begin{lemma}%
\label[lemma]{lem: Nader 8}
Assume that 
  $|\bX_i \cap J|\le 1$ for all $i\in [r]$. Then the algorithm always reaches line 29 and outputs ``accept.''
\end{lemma}
\begin{proof}
    By \Cref{lem: Nader7}, \Cref{algo:JuntaTester} reaches Phase~D. To show that it reaches the last line,  we need to show that it does not halt on line {23} or {27}. %
    
    For every $\ell \in I$ we have that $\bX_\ell$ contains exactly one relevant variable in $J$, which we denote by $x_{\tau(\ell)}$. So we have that $f^{\ell}$ is a literal; in particular, either $f^{\ell}(x) = x_{\tau(\ell)}$ or  $f^{\ell}(x) = \overline{x_{\tau(\ell)}}$. Since $\bY_{\ell,0}, \bY_{\ell,1}$ is a partition of $\bX_\ell$,  $\tau(\ell)$ is in exactly one of $\bY_{\ell,0}$ or $\bY_{\ell,1}$. Suppose w.l.o.g.~that  $\tau(\ell) \in \bY_{\ell,0}$, then for any $b = b^0 \circ b^1$ with $b^i \in \{0,1\}^{\bY_{{\ell,i}}}$ we have  
    $$f^\ell(b^{0} \circ b^{1}) \neq f^\ell(\overline{b^{0}} \circ {b^{1}})\quad\text{and}\quad
     f^\ell(b^{0} \circ b^{1})= f^\ell(b^{0} \circ \overline{b^{1}}).$$
Therefore, we always have $\bG_{\ell,0}=h$ and $\bG_{\ell,1}=0$. 
    Thus, the algorithm reaches line ${26}$.
    
Next, observe that if $\tau(\ell) \in \bY_{\ell,0}$, then we set $\bz_{\bX_\ell}=\bw_\ell$ and thus, $\bz_{\tau(\ell)}=\bw_{\tau(\ell)}=0$ (and if we had $\tau(\ell) \in \bY_{\ell,1}$, then we would set $\bz_{\bX_\ell}=\overline{\bw_\ell}$ and thus,  $\bz_{\tau(\ell)}=\overline{\bw_{\tau(\ell)}}=0$).

    So, for every $\ell \in I$ we have $\bz_{\tau(\ell)}=0$ and $\tau(\ell)$ is the only relevant variable in $\bX_\ell$. So $\bu_{\bX}$ and $\bu_{\bX} + \bz_{\bX}$ have the same value on all the relevant variables $J$ in $\bX$. As such we must always have $f(\bu_{\bX} \circ \bv)=f((\bu_{\bX} \oplus \bz_{\bX}) \circ \bv)$ so the algorithm reaches the last line and accepts. %
\end{proof}

\begin{lemma}
    If $f$ is a $k$-junta then  \Cref{algo:JuntaTester}  outputs ``accept'' with probability at least $2/3$.
\end{lemma}
\begin{proof}
    The result follows from \Cref{lem: Nader 6}, \Cref{lem: Nader7}, and \Cref{lem: Nader 8}.
\end{proof}

\subsection{The analysis when $f$ is $\eps$-far in relative distance from every $k$-junta} \label{subsec:farfromjunta}

Throughout this subsection we assume that $\reldist(f,g) > \eps$ for every $k$-junta $g: \zo^n \to \zo$.
Our goal is to show that in this case, \Cref{algo:JuntaTester} rejects with high probability.

For Phase A of the algorithm, we just fix any $r$-way 
  partition $X_1,\ldots,X_r$ of $[n]$.Let us consider Phase B of the algorithm.  
It is clear that if \Cref{algo:JuntaTester} passes Phase B (instead of being rejected on line {7}),
  then at the moment when it reaches Phase C we must have $t=T=O(\log (k/\eps))$ and $|I|\le k$.
Furthermore, we now show that the following condition is unlikely to happen:

\begin{lemma}%
\label[lemma]{lem: Nader 10}
The probability of  \Cref{algo:JuntaTester} reaching Phase C with 
$$\Prx_{\substack{\bu \sim f^{-1}(1)\\ \bw \sim \{0,1\}^{\overline \bX}}}\big[f(\bu_{\bX} \circ \bw) \neq f(\bu)\big] \geq \epsilon/20 $$ is at most $1/15$.
\end{lemma}
\begin{proof}
For \Cref{algo:JuntaTester} to reach Phase~C,
  it must be the case that it has drawn $\bu\sim f^{-1}(1)$ and $\smash{\bw\sim \{0,1\}^{\overline{\bX}}}$ for $T$ times and they satisfy $f(\bu_{\bX}\circ \bw) =f(\bu)$ every time.
But given the condition above, the probability that this happens is at most 
    $(1-\epsilon/20)^{T}$, which is at most $\frac{1}{15k}$
by setting the hidden constant in $T=O(\log(k/\eps))$ 
  sufficiently large.
By a union bound (over the at most $k$ many
  different $\bX$'s that are in play across all the executions of line~{5}), 
  the probability of  \Cref{algo:JuntaTester}
  reaching Phase~C with the condition above is at most $1/15$.
\end{proof}

Now assume that  \Cref{algo:JuntaTester} reaches Phase C with an $X$ that satisfies  
\begin{equation}\label{temp11}
\Prx_{\substack{\bu \sim f^{-1}(1)\\ \bw \sim \{0,1\}^{\overline X}}}\big[f(\bu_X \circ \bw) \neq f(\bu)\big] \le \epsilon/20.
\end{equation}

\begin{lemma}%
\label[lemma]{lem: Nader 11}If for some $\ell \in I$, $f^{\ell}:\{0,1\}^{X_\ell}\rightarrow \{0,1\}$ is $(1/30)$-far from every literal with respect to the uniform distribution, then \Cref{algo:JuntaTester} rejects in Phase~C with probability at least $14/15$.%
\end{lemma}

\begin{proof}
    If $f^{\ell}$ is $(1/30)$-far from every literal with respect to the uniform distribution then it is either \begin{enumerate}
        \item $(1/30)$-far from every $1$-junta,
        \item or $(1/30)$-far from every literal but $(1/30)$-close to a $0$-junta (i.e., a constant-$0$ or $1$ function).
    \end{enumerate}
    In case $1$, by \Cref{thm: uniform junta blais}, with probability at least $14/15$, UniformJunta($f^\ell, 1, 1/30, 1/15$) rejects, in which case \Cref{algo:JuntaTester} rejects. In case $2$,  $f^{\ell}$ is $(1/30)$-close to a constant function; say w.l.o.g.~that it is close to the constant-$0$ function. The probability we fail to reject on line {$12$} is :

    \begin{align*}
        \Prx_{\bb \sim \{0,1\}^{X_\ell}}\big[f^\ell(\bb) \neq f^\ell(\overline{\bb})\big] %
         \leq \Prx_{\bb \sim \{0,1\}^{X_\ell}}\big[ f^\ell(\overline{\bb}) =1\big]  + \Prx_{\bb \sim \{0,1\}^{X_\ell}}\big[f^\ell(\bb)=1\big]  
         \leq 1/15,
    \end{align*}
    where the last inequality follows from the fact $f^{\ell}$ is $(1/30)$-close to the constant-$0$ function under the uniform distribution and $\bb$, $\overline{\bb}$ are both sampled uniformly at random.
\end{proof}

Next assume that  \Cref{algo:JuntaTester} reaches Phase D
  with $X$ satisfying \Cref{temp11} and $I$ satisfying that $f^\ell$ is $(1/30)$-close to
  a literal with respect to the uniform distribution for 
  every $\ell\in I$.
For each $\ell\in I$, let $\tau(\ell)\in X_\ell$ be the (unique) variable such that $f^\ell$ is $(1/30)$-close to either 
  $x_{\tau(\ell)}$ or $\overline{x_{\tau(\ell)}}$. 

\begin{lemma}%
\label[lemma]{lem: Nader 12}
For each execution of the inner loop on $\ell$ in Phase D (i.e., lines {17--24}), 
  the probability that $\bz_{\tau(\ell)}$ is set to 1 on line {24} (in which case the algorithm does not reject on line {23}) is at most $(1/15)^h$.
\end{lemma}
\begin{proof}
    Fix some $\ell \in I$, and assume without loss of generality that $f^{\ell}$ is $(1/30)$-close to $x_{\tau(\ell)}$ under the uniform distribution. The case where $\smash{f^\ell}$ is close to $\overline{x_{\tau(\ell)}}$ is similar. Let $\bY_{\ell,0}, \bY_{\ell,1}$ be the partition of $X_\ell$ drawn in the loop. Then $\tau(\ell)$ is in exactly one of $\bY_{\ell,0}$ or $\bY_{\ell,1}$. 
Assume without loss of generality that {$\bs_{\tau(\ell)}=1$ so} $\tau(\ell) \in \bY_{\ell,1}$; the case when $\tau(\ell)\in \bY_{\ell,0}$ is similar. We have that  $\bz_{\tau(\ell)}=1$ only if we set $\bz_{\bX_\ell} \leftarrow \bs_{\bX_\ell}$, meaning  $\bG_{\ell,0}=h$ is a necessary condition for $\bz_{\tau(\ell)}=1$. So, it suffices to show that  
$\Pr[\bG_{\ell,0}=h] \leq (1/15)^h$. 

Focusing on a single execution of the loop spanning lines {19--21},
  we have%
$$
\Pr\big[f^\ell(\bb^0\circ\bb^1)\ne f^\ell(\overline{\bb^0}\circ\bb^1)\big]\le 
\Pr\big[f^\ell(\bb^0\circ\bb^1)\ne \bb^1_{\tau(\ell)}\big]+
\Pr\big[f^\ell(\overline{\bb^0}\circ {\bb^1})\ne \bb^1_{\tau(\ell)}\big]
\le 1/15,
$$
where the second inequality follows from the assumption that $f^\ell$ is $(1/30)$-close to $x_{\tau(\ell)}$ under the uniform distribution and the marginal distribution of both $\smash{\bb^0\circ\bb^1}$ and $\smash{\overline{\bb^0}\circ\bb^1}$ is uniform.
It follows that $\Pr[\bG_{\ell,0}=h]\leq (1/15)^h$.
\end{proof}

We are now ready to show that the algorithm accepts with
  probability at most $1/3$. 

\begin{lemma}%
\label[lemma]{lem:last-lemma-no-case}
    If $\reldist(f,g)>\eps$ for every $k$-junta $g$, then \Cref{algo:JuntaTester}  rejects
    with probability at least $2/3$.
\end{lemma}
\begin{proof}
By \Cref{lem: Nader 10} and \Cref{lem: Nader 11}, we have that with probability at least $13/15$, \Cref{algo:JuntaTester} either already rejects in Phases A--C or reaches Phase~D with $X$ and $I$ satisfying \Cref{temp11} and $f^\ell$ being $(1/30)$-close to
  a literal (either $x_{\tau(\ell)}$ or $\overline{x_{\tau(\ell)}}$) with respect to the uniform distribution for 
  every $\ell\in I$.
Assume that the latter happens. 
We show below that Phase~D of \Cref{algo:JuntaTester} rejects with probability at least $13/15$.
It follows from a union bound that overall
  \Cref{algo:JuntaTester} accepts with probability
  at most $2/15+2/15<1/3$.

To this end, by \Cref{lem: Nader 12} and a union bound,
  with probability at least  $1-M|I|(1/15)^h $, 
  either \Cref{algo:JuntaTester} rejects on some execution of line {23} 
  or $\bz_{X_\ell}$ satisfies $\bz_{\tau(\ell)}=0$ for every $\ell$
  and in every one of the $M$ loops.
By making the hidden constant in $h$ sufficiently larger
  than the hidden constant in $M$,
  we can make $$1-M|I|(1/15)^h\ge 1-O(k/\eps)\cdot (1/15)^h\ge 14/15. $$
  
Assuming that this happens and \Cref{algo:JuntaTester} 
  does not reject on line {23}, 
  then for each of the $M$ loops, we claim that the joint distribution of 
  $(\bu_X\circ \bw,(\bu_X\oplus \bz_X)\circ \bw)$ is exactly the same as
  the joint distribution of $(\bu_X\circ \bw,\bu_J\circ \by\circ  \bw)$
  in \Cref{lem:keyjunta} %
  with $J=\{\tau(\ell):\ell\in I\}$. 
  To see this, first observe that the joint distributions of the $\overline{X}$ coordinates (corresponding to $\bw$) are clearly the same in both cases, so it suffices to consider coordinates in $X$. We observe that for any $\tau(\ell) \in J$, we have $\left((\bu_X \oplus \bz_X) \circ \bw \right)_{\tau(\ell)}=\bu_{\tau(\ell)} \oplus \bz_{\tau(\ell)}= \bu_{\tau(\ell)}$, since $\bz_{\tau(\ell)}=0$ by assumption. Turning to the coordinates in $X \setminus J$, fix any $i \in X_\ell, \ell \in I$ with $i \neq \tau(\ell)$. We have $\left((\bu_X \oplus \bz_X) \circ \bw \right)_{i}=\bu_i \oplus \bz_i$, but $\bz_i=\bs_i$ if $\bs_{\tau(\ell)}=0$ and $\bz_i=\overline{\bs_i}$ otherwise. Since $\bs_i \sim \zo$ and is independent of $\bs_{\tau(\ell)}$, we thus have that $\bz_i \sim \zo$ as well. So $\left((\bu_X \oplus \bz_X) \circ \bw \right)_{i}$ is uniformly random, which matches the distribution of $(\bu_J \circ \by \circ w)_i$ as desired.

As a result,  \Cref{algo:JuntaTester} rejects on line {27} with 
  probability at least $1-(1-\eps/5)^M\ge 14/15$ by making the constant
  hidden in $M$ sufficiently large. 
Therefore,  \Cref{algo:JuntaTester} rejects in Phase~D with probability at least $13/15.$ This finishes the proof of the lemma.
\end{proof}

\end{document}